\documentclass[a4paper,twoside,11pt]{article}
\usepackage{amsfonts,amsmath,amssymb,amsthm}
\usepackage{geometry}
\usepackage{booktabs}  

\usepackage[normalem]{ulem}
\usepackage{hyperref}
\hypersetup{
    colorlinks,%
    citecolor=blue,%
    filecolor=black,%
    linkcolor=blue,%
    urlcolor=blue
}
\usepackage{natbib}

\providecommand{\keywords}[1]{{\textit{Keywords:}} #1}

\RequirePackage{authblk}

\usepackage{amssymb,color}
\usepackage{graphicx}
\usepackage{natbib}
\usepackage{graphicx}
\usepackage{subfigure}
\usepackage{blkarray}
\usepackage{float}
\DeclareMathOperator*{\argmax}{argmax}

\newcommand{\pkg}[1]{{\fontseries{b}\selectfont #1}}
\let\proglang=\textsf

\newcommand{\ty}{\boldsymbol{y}}
\newcommand{\tv}{\boldsymbol{v}}
\newcommand{\tx}{\boldsymbol{x}}
\newcommand{\tz}{\boldsymbol{z}}
\newcommand{\tY}{\boldsymbol{Y}}
\newcommand{\tX}{\boldsymbol{X}}
\newcommand{\tZ}{\boldsymbol{Z}}
\newcommand{\tV}{\boldsymbol{V}}
\newcommand{\bA}{\boldsymbol{A}}
\newcommand{\bP}{\boldsymbol{P}}
\newcommand{\bpi}{\boldsymbol{\pi}}
\newcommand{\bvarepsilon}{\boldsymbol{\varepsilon}}
\newcommand{\blambda}{\boldsymbol{\lambda}}

\newcommand{\btheta}{\boldsymbol{\theta}}

\newcounter{thm} 
 
\newcounter{lm} 
\newcounter{co} 
\newcounter{assum}
\newcounter{rem} 

\newtheorem{theorem}[thm]{Theorem}
\newtheorem{lemma}[lm]{Lemma}
\newtheorem{corollary}[co]{Corollary}
\newtheorem{assumption}[assum]{Assumption}
\newtheorem{remark}[rem]{Remark}

\title{Mixture of hidden Markov models for accelerometer data}
\author{
{Marie du Roy de Chaumaray\thanks{Marie du Roy de Chaumaray\\
\hspace*{1.8em}CREST, ENSAI, France, 
E-mail: marie.du-roy-de-chaumaray@ensai.fr
}
, Matthieu Marbac\thanks{Matthieu Marbac\\
\hspace*{1.8em}CREST, ENSAI, France,
E-mail: matthieu.marbac-lourdelle@ensai.fr}
, Fabien Navarro\thanks{Fabien Navarro\\
\hspace*{1.8em}CREST, ENSAI, France,
E-mail: fabien.navarro@ensai.fr}}
}
\date{\today}

\usepackage{fancyhdr}
\begin{document}

\maketitle

\begin{abstract}
Motivated by the analysis of accelerometer data,  we introduce a specific finite mixture of hidden Markov models with particular characteristics that adapt well to the specific nature of this type of data. Our model allows for the computation of statistics that characterize the physical activity of a subject (\emph{e.g.}, the mean time spent at different activity levels and the probability of the transition between two activity levels) without specifying the activity levels in advance but by estimating them from the data. In addition, this approach allows the heterogeneity of the population to be taken into account and subpopulations with homogeneous physical activity behavior to be defined. 
We prove that, under mild assumptions, this model implies that the probability of misclassifying a subject decreases at an exponential decay with the length of its measurement sequence. Model identifiability is also investigated. We also report a comprehensive suite of numerical simulations to support our theoretical findings. Method is motivated by and applied to the PAT study.
\end{abstract}

\noindent\keywords Accelerometer data; Hidden Markov model; Longitudinal data; Missing data; Mixture models

\section{Introduction}

Inadequate sleep and physical inactivity affect physical and mental well-being while often exacerbating health problems. They are currently considered major risk factors for several health conditions (see, for instance \citet{kimm2005relation,taheri2004short,lee2012effect,grandner2013sleep,mctiernan2008mechanisms}).
Therefore, appropriate assessment of activity and sleep periods is essential in disciplines such as medicine and epidemiology. The use of accelerometers to evaluate physical activity---by measuring 
the acceleration of the part of the body to which they are attached---is a classic method that has become widespread in public health research. Indeed, since the introduction in 2003 of the first objective assessment of physical activity using accelerometers, as part of the National Health and Nutrition Examination Survey (NHANES), the analysis of actigraphy data has been the subject of extensive studies over the past two decades. Recently, the New York City (NYC) Department of Health and Mental Hygiene conducted the 2010-2011 Physical Activity and Transit (PAT) Survey\footnote{NYC Department of Health and Mental Hygiene. Physical Activity and Transit Survey 2010-2011; public use datasets accessed on May 10, 2019. The data are freely accessible on this page: https://www1.nyc.gov/site/doh/data/data-sets/physical-activity-and-transit-survey-public-use-data.page}, a random survey of adult New Yorkers that tracked levels of sedentary behavior and physical activity at work, at home, and for leisure. A subset of interviewees was also invited to participate in a follow-up study to measure objectively their activity level using an accelerometer. One of the objectives of this study is to describe measured physical activity levels and to compare estimates of adherence to recommended physical activity levels, as assessed by accelerometer, with those from self-reports. In contrast to NHANES accelerometer data, PAT data still seem relatively unexplored in the statistical literature.

This paper is motivated by the analysis of the accelerometer data worn by 133 individuals aged at least of 65 years  who responded PAT survey. Our objective is to propose a model adapted to the specificities of these data and study its properties. Indeed, this data set raises various challenges, such as managing the heterogeneity of the population or missing data of different natures. In order to motivate the development of a new model, we present an overview of the literature on accelerometer data analysis.

Pioneer approaches used for analyzing accelerometer data have focused on automatic detection of the sleep and wake-up periods \citep{cole1992automatic,sadeh1994activity,pollak2001accurately,van2015novel}. More recent developments are interested in the classification of different levels of activity (see \citet{yang2010review} for a review). These methods provide summary statistics like the mean time spent at different activity levels. In epidemiological studies, time spent by activity level is often used as a covariate in predictive models (see, for instance, the works of \citet{noel2010use,PALTA2015824,innerd2018using} where the links between physical activity and obesity are investigated). These statistics can be computed using deterministic cutoff levels \citep{freedson1998calibration}. However, with such an approach, the dependency in time is neglected and the cutoff levels are pre-specified and not estimated from the data.
 
Accelerometer data are characterized by a time dependency between the different measures. 
They can be analyzed by methods developed for functional data or by Hidden Markov Models (HMM).
Methods for functional data need the observed data to be converted into a function of time \citep{morris2006using,xiao2014quantifying,gruen2017use}. For instance, \citet{morris2006using} use wavelet basis for analyzing accelerometer profiles. The use of a function basis reduces the dimension of the data, and therefore the computing time.  However, these methods do  not define levels of activity and thus cannot directly provide the time spent at different activity levels.

When considering a discrete latent variable to model time dependence, HMM are appropriate for adjusting sequence data \citep{scott2005hidden,altman2007mixed,GassiatStatCo2016}. \citet{titsias2016statistical} expand the amount of information which can be obtained from HMM including a procedure to find maximum \emph{a posteriori} (MAP) of the latent sequences and to compute posterior probabilities of the latent states. HMM are used on activity data for monitoring circadian rythmicity \citep{huang2018hidden} or directly for estimating the sequence of activity levels from accelerometer data \citep{witowski2014using}. For simulated data, \cite{witowski2014using} established the superiority of different HMM models, in terms of classification error, over traditional methods based on \emph{a priori} fixed thresholds. While the simplicity of implementing threshold-based methods is an obvious advantage, they have some significant disadvantages compared to the HMM methods, particularly for real data. Indeed, the variation in counts and the resulting dispersion is large, leading to considerable misclassification of counts recorded in erroneous activity ranges. The approach of \citet{witowski2014using} assumes  homogeneity of the population and does not consider missingness within the observations. However, heterogeneity in physical activity behaviors is often present (see, for instance, \citet{geraci2018additive}) and the use of more than one HMM allows it to be taken into account (see, \emph{e.g.}, \cite{van1990mixed}).
Clustering enables the heterogeneity of the population to be addressed by grouping observations into a few homogeneous classes. Finite mixture models \citep{McLachlan:04, McNicholas2016} permit to cluster different types of data like: continuous \citep{Ban93}, integer \citep{karlis2007finite}, categorical \citep{Goo74}, mixed \citep{Hunt:11, kosmidis2015mixture}, network \citep{HoffJasa2002, HunterJasa2008,Matias2018} and sequence data \citep{wong2000mixture}. 
Recent methods use clustering for accelerometer data analysis. 
For instance, \citet{wallace2018variable} use a specific finite mixture  to identify
novel sleep phenotypes, \citet{huang2018multilevel} perform a matrix-variate-based clustering on accelerometer data while \citet{lim2019functional} use a clustering technique designed for functional data.
Mixed Hidden Markov Models (MHMM) are a combination of HMM and Generalized Linear Mixed Models \citep{van1990mixed,bartolucci2012latent}. These models consider one (or more) random effect(s) coming from either a continuous distribution \citep{altman2007mixed} or a discrete distribution \citep{bartolucci2011assessment,maruotti2011mixed}. Note that a MHMM with a single discrete random effect distribution, having a finite number of states, is a finite mixture of HMM. Such a model allows to estimate a partition among the population and to consider the population heterogeneity. The impact of the random effect can be on the measurement model or on the latent model.

This paper focuses on the analysis of PAT data with a two-fold objective: obtain summary statistics about physical activity of the subjects without pre-specifying cutoff levels and obtain a partition which groups subjects into homogeneous classes. We define  a class as homogeneous if its subjects have similar average times spent into the different activity levels and similar transition probabilities between activity levels. To achieve this goal, we introduce a specific finite mixture of HMM for analyzing accelerometer data. This model considers two latent variables: a categorical variable indicating each subject's class membership and a sequence of categorical variables indicating the subject's level of activity each time its acceleration is measured. At time $t$, the measure is independent of the class membership conditionally on the activity level (\emph{i.e.}, the latent state) and follows a zero-inflated distribution---a distribution that allows for frequent zero-valued observations. The activity level defines the parameter of this distribution. The use of zero-inflated distribution is quite common for modeling accelerometer data  \citep{ae2018missing,bai2018two}, as the acceleration is measured every second, many observations are zero. Note that the definitions of the activity levels are equal among the mixture components. This is an important point for the use of the summary statistics (\emph{e.g.}, time spent at different activity levels, probabilities of transition between levels) in a future statistical study. The model we consider is thus a specific MHMM with a finite-states random effect that only impacts the distribution of latent physical activity levels. MHMM with a finite-states random effect have few developments in the literature \citep{bartolucci2011assessment,maruotti2011mixed}, especially when the random effects only impact the latent model (and not the measurement model). We propose to theoretically study the model properties by showing that the probability of misclassifying an observation decreases at an exponential rate. In addition, since the distribution given the latent state is itself a bi-component mixture (due to the use of zero-inflated distributions), we investigate sufficient conditions for model identifiability.

In practice, the data collected often include missing intervals due to non-compliance by participants (\emph{e.g.}, if the accelerometer is removed). 
\citet{geraci2016probabilistic} propose to identify different profiles of physical activity behaviors using a principal component analysis that allows for missing values. The PAT data contain three types of missing values corresponding to periods when the accelerometer is removed, making statistical analysis more challenging. First, missingness occurs at the beginning and at the end of the measure sequences due to the installation and the removing of the accelerometer. Second, subjects are asked to remove the accelerometer when they sleep at night. Third, missing values appear during the day (\emph{e.g.}, due to a shower period, nap, ...). We remove missing values which occur at the begin and at the end of the sequence. For missingness caused by night sleep, we consider that the different sequences describing different days of observations of a subject are independent and that the starting point (\emph{e.g.}, first observed measure of the accelerometer of the day) is drawn from the stationary distribution. For missing values measured during the day, the model and the estimation algorithm can handle these data. Moreover, we propose an approximation of the distribution that avoids the computation of large powers of the transition matrices in the algorithm used for parameter inference and thus reduces computation time. Theoretical guarantees and numerical experiments show the relevance of our proposition.

The \proglang{R} package \pkg{MHMM} which implements the method introduced in this paper is available on CRAN \citep{MHMM}. It permits to analyze other accelerometer data and thus it is complementary to existing packages for MHMM. Indeed, it takes into account the specificities of accelerometer data (the class membership only impacts the transition matrices, the emission distributions are zero-inflated gamma (ZIG) distributions). Among the \proglang{R} packages implementing MHMM methods, one can cite the \proglang{R} packages \pkg{LMest} \citep{bartolucci2017lmest}  and \pkg{seqHMM} \citep{seqHMM} which focus on univariate longitudinal categorical data and the \proglang{R} package \pkg{mHMMbayes} \citep{mHMMbayes} which focuses on multivariate longitudinal categorical data.

This paper is organized as follows. Section~\ref{sec:data} presents the PAT data and the context of the study. Section~\ref{sec:model} introduces our specific mixture of HMM and its justification in the context of accelerometer data analysis. Section~\ref{sec:properties} presents the model properties (model identifiability, exponential decay of the probabilities of misclassification and a result for dealing with the non-wear periods). Section~\ref{sec:MLE} discusses the maximum likelihood inference and Section~\ref{sec:simu} illustrates the model properties on both simulated and real data. Section~\ref{sec:appli} illustrates the approach by analyzing a subset of the PAT accelerometer data. Section~\ref{sec:concl} discusses some future developments. Proofs and technical lemmas are postponed in Appendix.

\section{PAT data description} \label{sec:data}
We consider a subset of the data from the PAT survey, the subjects who participated in the follow-up study to measure objectively their activity level using an accelerometer. A detailed methodological description of the study and an analysis of the data is provided in \cite{immerwahr2012physical}. Note that the protocols for accelerometer data for the PAT survey and NHANES were identical. One of the objectives of the PAT study is to investigate the relationships  between self-reported physical activity and physical activity measured by the accelerometer in order to provide best practice recommendations for the use of self-reported data \cite{wyker2013self}. Indeed, self-reported data may be subject to overreporting. This is particularly the case among less active people, due in particular to a social desirability bias or the cognitive challenge associated with estimating the frequency and duration of daily physical activity (see, \emph{e.g.}, \cite{slootmaker2009disagreement,dyrstad2014comparison,lim2015measurement}). The results of \cite{wyker2013self} show that males tend to underreport their physical activity, while females and older adults (65 years and older) overreported (see also \cite{troiano2008physical} for a detailed study of the differences between self-reported physical activity and accelerometer measurements in NHANES 2003-2004). Consequently, the study of data measured by accelerometer for these specific populations makes it possible to determine methods for correcting estimates from self-reported data, such as stratification by gender and/or age when comparing groups.

\begin{figure}[h!t]
\centering 
\includegraphics[width=\textwidth]{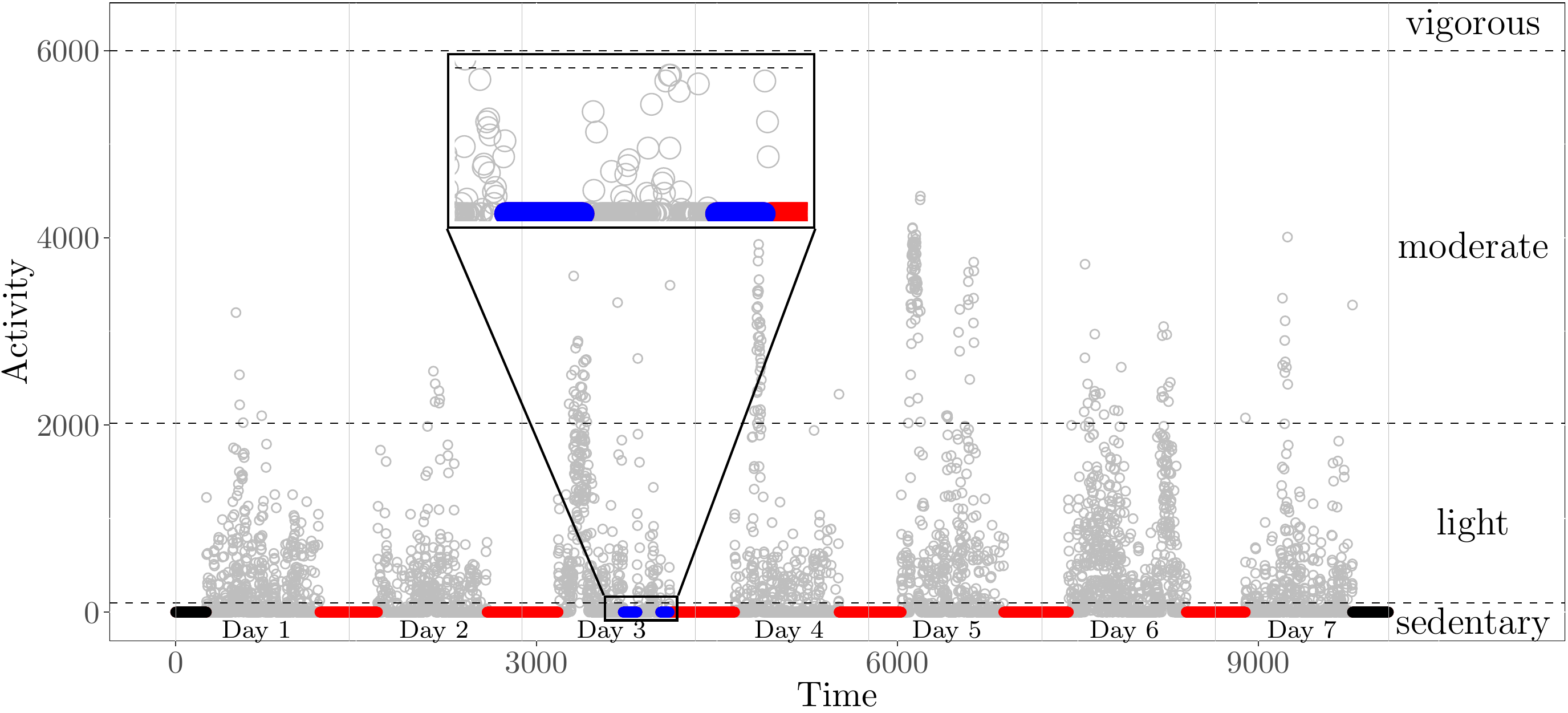}
\caption{Accelerometer data of subject Patcid:1200255 of the PAT study measured for one week (with a zoom on the afternoon of day 3): observed values (in gray), missing values during a daytime period (in blue), missing values during a period of night sleep (in red) and missing values at the start and end of the measurement period (in black). The dashed horizontal lines represent the four levels of physical activity based on the classification established by the \cite{us20082008}.}
\label{fig:example}
\end{figure}

In this work, we are particularly interested in the age category above 65 years old ($n=133$). 
We present some characteristics related to PAT data and refer to \cite{immerwahr2012physical} for a full description\footnote{Raw accelerometer data, covariates allowing the selection of the subset of the population, as well as a detailed dictionary are freely accessible here: https://www1.nyc.gov/site/doh/data/data-sets/physical-activity-and-transit-survey-public-use-data.page}. Accelerometers were worn for one week (beginning on Thursday and ending on Wednesday) and measured the activity minute-by-minute. The trajectory associated with each subject is therefore of length 10080.  In addition, a participant's day spans from 3am-3am (and not a calendar day) in order to record late night activities and transit and contains missing data sequences of variable length at the beginning and end of the measurement period (these missing data sequences were excluded from the analysis). This length varies from one subject to another, and the mean and minimum trajectory length for the population under consideration (after excluding those missing at the edges) are 9474 and 5199 respectively (with a total number of observations equal to 1259981). The model of the accelerometer used is Actigraph GT3X, it is worn on the hips (which results in the fact that certain activities, such as lifting weights or biking, cannot be measured). In addition, participants were also asked to remove it when sleeping, swimming or bathing, hence the data contains approximately 44$\%$ of missing values that appear mainly in sequence, appearing at night but also during the day. Figure~\ref{fig:example} gives an example of accelerometer data measured on one subject (\emph{i.e.}, patcid:1200255) for one week where the three types of missing data can be seen. The four levels of physical activity based on the classification established by the \cite{us20082008} in the Physical Activity Guidelines for Americans (PAGA) report is also shown in the Figure~\ref{fig:example}. Specifically, the PAT protocol for accelerometer data has established a classification according to PAGA, characterizing each minute of activity. Activity minutes with less than 100 activity counts were classified as Sedentary, minutes with 100-2019 counts were classified as Light, the class Moderate corresponds to 2020-5998 counts/minute and Vigorous 5999 and above counts/minute. A comparison between our method and this traditional threshold-based approach is provided in Section~\ref{subsec:level}.

\section{Mixture of hidden Markov models for accelerometer data} \label{sec:model}
In this section we present the proposed model and the application context for which it has been defined.

\subsection{The data}
Observed data $\ty=(\ty_1^\top,\ldots,\ty_n^\top)$ are composed of $n$ independent and identically distributed  sequences $\ty_i$. Each sequence $\ty_i=(y_{i(0)},\ldots,y_{i(T)})$ which contains the values measured by the accelerometer at times $t\in\{0,1,\ldots, T\}$ for subject $i$, with $y_{i(t)}\in\mathbb{R}^+$. Throughout the paper, index $i$ is related to the label of the subject and index $(t)$ is related to the time of measurement.

The model considers $M$ different activity levels (which are unobserved). These levels impact the distribution of the observed sequences of accelerometer data. The sequence of the hidden states $\tx_i$ indicates the activity level of subject $i$ at the different times. Thus, $\tx_i=(\tx_{i(0)},\ldots,\tx_{i(T)})\in\mathcal{X}$ and the activity level (among the $M$ possible levels) of subject $i$ at time $t$ is defined by the binary vector $\tx_{i(t)}=(x_{i(t)1},\ldots,x_{i(t)M})$ where $x_{i(t)h}=1$ if subject $i$ is at state $h$ at time $t$ and $x_{i(t)h}=0$ otherwise.

The heterogeneity (in the sense of different physical activity behaviors) between the $n$ subjects can be addressed by grouping subjects into $K$ homogeneous classes. This is achieved by clustering that assesses a partition $\tz=(\tz_1,\ldots,\tz_n)$ among the $n$ subjects based on their accelerometer measurements. Thus the vector $\tz_{i}=(z_{i1},\ldots,z_{iK})$ indicates the class membership of subject $i$, as $z_{ik}=1$ if observation $i$ belongs to class $k$ and $z_{ik}=0$ otherwise. Throughout the paper, index $k$ refers to the label of a class grouping homogeneous subjects.

Each subject $i$ is described by three random variables: one unobserved categorical variable $\tz_i$ (defining the membership of the class of homogeneous physical activity behaviors for subject $i$), one unobserved categorical longitudinal data $\tx_i$ (a univariate categorical discrete-time time series which defines the activity level of subject $i$ at each time) and one observed positive longitudinal data $\ty_i$ (a univariate positive discrete-time time series which contains the values of the accelerometer measured on subject $i$ at each time).

\subsection{Generative model}
\begin{figure}[ht!]
\begin{center}
\includegraphics[width=0.9\textwidth]{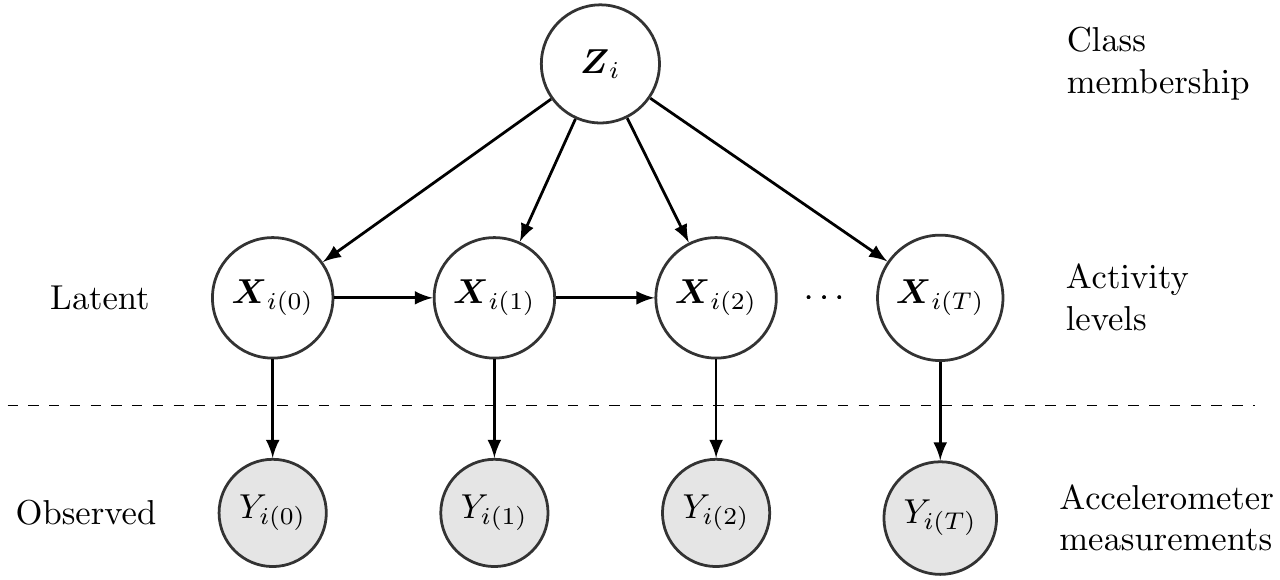}
\caption{Generative model of the specific mixture model of HMM used for the accelerometer data:  an arrow between two variables indicates dependency and an absence of arrow indicates conditional independence.}
\label{fig:generative}
\end{center}
\end{figure}
The model described below considers that the observations are independent between the subjects and identically distributed.  It is defined by the following generative model and summarized by Figure~\ref{fig:generative} (note that this figure is similar to Figure~6.2 of \citet{bartolucci2012latent}):
\begin{enumerate}
\item sample class membership $\tz_i$  from a multinomial distribution;
\item sample the sequence of activity levels  $\tx_i$ from a Markov model whose transition matrix depends on class membership;
\item sample the accelerometer measurement sequence given the activity levels (each $\tY_{i(t)}$ follows a ZIG distribution whose parameters are defined only by $\tx_{i(t)}$).
\end{enumerate}

\subsection{Finite mixture model for heterogeneity}
The sequence of accelerometer measures obtained on each subject is assumed to independently arise from a mixture of $K$ parametric distributions, so that the probability distribution function (pdf) of the sequence $\ty_i$ is
\begin{equation}
p(\ty_i; \btheta ) =\sum_{k=1}^K \delta_k \, p(\ty_i ; \bpi_{k}, \bA_k, \blambda, \bvarepsilon), \label{eq:model}
\end{equation}
where $\btheta=\{\blambda,\bvarepsilon\}\cup\{\delta_k,\bpi_k,\bA_k;k=1,\ldots,K\}$ groups the model parameters, $\delta_k$ is the proportion of components $k$ with $\delta_k>0$, $\sum_{k=1}^K \delta_k=1$, and $p(\cdot ; \bpi_{k}, \bA_k, \blambda, \bvarepsilon)$ is the pdf of component $k$ parametrized  by $(\bpi_{k}, \bA_k, \blambda, \bvarepsilon)$ defined below. Thus, $\delta_k$ is the marginal probability that a subject belongs to class $k$ (\emph{i.e.}, $\delta_k=\mathbb{P}(Z_{ik}=1)$). Moreover, $p(\cdot ; \bpi_{k}, \bA_k, \blambda, \bvarepsilon)$ defines the distribution of a sequence of values measured by the accelerometer on a subject belonging to class $k$ (\emph{i.e.},  $p(\cdot ; \bpi_{k}, \bA_k, \blambda, \bvarepsilon)$ is the pdf of $\ty_i$ given $Z_{ik}=1$).

\subsection{Hidden Markov model for activity levels} 
The model assumes that the distribution of the hidden state sequence depends on the class membership, and that the distribution of activity measurements depends on the state at time $t$ but not on the component membership given the state (\emph{i.e.}, $\tX_i \not \perp \tZ_i$,  $Y_{i(t)} \not \perp \tX_{i(t)}$ and $Y_{i(t)}\perp \tZ_i \mid\tX_{i(t)}$). It is crucial that the distribution of $Y_{i(t)}$ given $\tX_{i(t)}$ is independent to $\tZ_i$. Indeed, each activity level is defined by the distribution of $Y_{i(t)}$ given the state. Therefore, to extract summary statistics on the whole population (as the average time spent by level of activity) the definition of the activity levels (and the distribution of $y_{i(t)}$ given the state) must to be the same among the mixture components.

The pdf of $\ty_i$ for components $k$ (\emph{i.e.}, given $Z_{ik}=1$) is
\begin{equation} \label{eq:compo}
p(\ty_i; \bpi_{k}, \bA_k, \blambda, \bvarepsilon) = \sum_{\tx_i \in \mathcal{X}} p(\tx_i; \bpi_{k}, \bA_k) \, p(\ty_i \mid \tx_i; \blambda, \bvarepsilon).
\end{equation}
The Markov assumption implies that
\begin{equation*}
p(\tx_i; \bpi_{k}, \bA_k) = \prod_{h=1}^\ell \pi_{kh}^{x_{i(0)h}} \prod_{t=1}^T\prod_{h=1}^M\prod_{\ell=1}^M (\bA_k[h,\ell])^{x_{i(t-1)h}x_{i(t)\ell}}, 
\end{equation*}
where  $\bpi_{k}=(\pi_{k1},\ldots,\pi_{kM})$ defines the initial probabilities so that $\pi_{kh}=\mathbb{P}(X_{i(1)h}=1 \mid Z_{ik}=1)$, $\bA_k$ is the  transition matrix so that $\bA_k[h,\ell]=\mathbb{P}(X_{i(t)\ell}=1 \mid X_{i(t-1)h}=1, Z_{ik}=1)$. 
Finally, we have 
\begin{equation*}
p(\ty_i \mid \tx_i; \blambda, \bvarepsilon) = \prod_{t=0}^T \prod_{h=1}^M g(y_{i(t)} ; \blambda_h, \varepsilon_h)^{x_{i(t)h}},
\end{equation*}
where $g(\cdot;\blambda_h, \varepsilon_h)$ is the pdf of a zero-inflated distribution defined by
\begin{equation*}
g(y_{i(t)} ; \blambda_h, \varepsilon_h) = (1-\varepsilon_h) g_c(y_{i(t)} ; \blambda_h) + \varepsilon_h \mathbf{1}_{\{y_{i(t)}=0\}},
\end{equation*}
where $g_c(\cdot ; \blambda_h)$ is the density of a distribution defined on a positive space and parametrized by $\blambda_h$. 
The choice of considering zero-inflated distributions is motivated by the large number of zero in the accelerometer data (see Figure~\ref{fig:example}). 
For the application of Section~\ref{sec:appli}, we use a gamma distribution of $g_c(\cdot ; \blambda_h)$. However, model properties and inference are discussed for a large family of densities $g_c(\cdot ; \blambda_h)$. 

\section{Model properties} \label{sec:properties}
In this section, we present the properties of the mixture of parametric HMM. We start with a discussion of three assumptions. Then, model identifiability is proved. It is shown that the probability of making an error in the partition estimation exponentially decreases with $T$, when the model parameters are known. Finally, the analysis of missing data is discussed. 
\subsection{Assumptions}
\begin{assumption}\label{ass:chaine}
For each component $k$, the Markov chain is irreducible. Moreover, we assume that the sequence is observed at its stationary distribution (\emph{i.e.}, $\bpi_k$ is the stationary distribution so $\bpi_k^\top\bA_k=\bpi_k^\top$). Therefore, there exists $0 \leq \nu<1$ such that
\begin{equation*}
\forall k\in\{1,\ldots,K\}, \; \nu_2(\bA_k)\leq\nu,
\end{equation*}
where $\nu_2(\bA_k)$ is the second-largest eigenvalue of $\bA_k$. Finally, we denote by $\bar\nu_2(\bA_k)=\max(0,\nu_2(\bA_k))$.
\end{assumption}

\begin{assumption}\label{ass:order}
The hidden states define different distributions for the observed sequence. Therefore, for $h\in\{1,\ldots,M\}$, $h'\in\{1,\ldots,M\}\setminus \{h\}$, we have $\blambda_h \neq \blambda_{h'}$. Moreover, the parametric family of distributions defining $g_c(\cdot ; \blambda_1),\ldots,g_c(\cdot ; \blambda_M)$ permits to consider an ordering such that for a fix value $\rho\in\mathbb{R}^+\setminus \{0\}$, we have
\begin{equation*}
\forall h \in\{1,\ldots, M-1\},\; \lim_{y_{i(1)} \to \rho} \frac{g_c(y_{i(1)} ; \blambda_{h+1})}{g_c(y_{i(1)} ; \blambda_h)} = 0.
\end{equation*}
\end{assumption}

\begin{assumption} \label{ass:kullback}
The transition probabilities are different over the mixture components and are not zero. Therefore, for $k\in\{1,\ldots,K\}$, $k'\in\{1,\ldots,K\} \setminus \{k\}$, we have $\forall (h,\ell)$, $\bA_k[h,\ell]\neq  \bA_{k'}[h,\ell]$. Moreover, there exists $\zeta>0$ such that
\begin{equation*}
\forall k\in\{1,\ldots,K\}, \; \forall k'\in\{1,\ldots,K\}\setminus\{k\},\;
\sum_{h=1}^M \sum_{\ell = 1}^M \pi_{kh} \log \frac{\bA_k[h,\ell]}{\bA_{k'}[h,\ell]} > \zeta.
\end{equation*}
Finally, without loss of generality, we assume that $A_k[1,1]>A_{k+1}[1,1]$.
\end{assumption}

Assumption~\ref{ass:chaine} states that the state at time 1 is drawn from the stationary distribution of the component that the observation belongs to. To obtain the model identifiability we do not need the assumption that the stationary distribution is different over the mixture components. As a result, two components having the same stationary distribution but different transition matrices can be considered. Assumption~\ref{ass:order} and Assumption~\ref{ass:kullback} are required to obtain the model identifiability. Assumption~\ref{ass:kullback} can be interpreted as the Kullback-Leibler divergence between the distribution of the states under component $k$ and their distribution under component $k'$. This constraint is required for model identifiability because it is related to the definition of the classes. Consequently, the matrices of the transition probability must be different among components.

\subsection{Identifiability}
Model identifiability is crucial for interpreting the estimators of the latent variables and of the parameters. 
It has been studied for some mixture models \citep{teicher1963,teicher1967,All09,CelisseEJS2012} and HMM \citep{GassiatStatCo2016}, but not for the mixture of HMM.
Generic identifiability (up to switching of the components and of the states) of the model defined in \eqref{eq:model} implies that 
\begin{equation*}
\forall \ty_i,\; p(\ty_i;\btheta) = p(\ty_i;\tilde \btheta) \Rightarrow \btheta = \tilde \btheta.
\end{equation*}
The following theorem states this property. 

\begin{theorem}\label{thm:ident}
If Assumptions \ref{ass:chaine}, \ref{ass:order} and \ref{ass:kullback} hold, then the model defined in \eqref{eq:model} is generically identifiable (up to switching of the components and of the states) if $T>2K$.
\end{theorem}
Proof of Theorem~\ref{thm:ident} is given in Appendix, Section~\ref{app:ident}.
The model defined by the marginal distribution of a single $y_{i(t)}$ is not identifiable. Indeed, the marginal distribution of $y_{i(t)}$ is a mixture of zero-inflated distributions and such mixture is not identifiable (\emph{i.e.}, different class proportions and inflation proportions can define the same distribution). It is therefore this dependency over time that makes the proposed mixture generically identifiable.
Note that such statement has been made by \citet{GassiatStatCo2016} when they discuss the case where the emission distribution for an HMM follows a mixture model.

\subsection{Probabilities of misclassification}
In this section, we examine the probability that an observation will be misclassified when the model parameters are known. 
We consider the ratio between the probability that subject $i$ belongs to class $k$ given $\ty_i$ and the probability that this subject belongs to its true class, and we quantify the probability of it being greater than some positive constant $a$. Let $\theta_0$ be the true model parameter and $\mathbb{P}_0=\mathbb{P}(\cdot \mid Z_{ik_0}=1, \theta_0)$ denote the true conditional distribution (true label of subject $i$ and parameters are known). 

\begin{theorem}\label{thm:conv}
Assume that Assumptions~\ref{ass:chaine} and \ref{ass:kullback} hold. If $a>0$ is such that Assumption~\ref{ass:expo} (defined in Appendix Section~\ref{app:errorclassif}) holds, then for every $k\neq k_0$
\begin{equation*}
\mathbb{P}_0\left[ \frac{\mathbb{P}(Z_{ik}=1 \mid \ty_i)}{\mathbb{P}(Z_{ik_0}=1 \mid \ty_i)} > a \right] \leq \mathcal{O}(e^{-cT}),
\end{equation*}
where $c>0$ is a positive constant
\end{theorem}

Moreover, the exponential bounds of Theorem~\ref{thm:conv} allows to use the Borel-Cantelli's lemma to obtain the almost sure convergence. 
\begin{corollary}\label{cor:conv}
Assume that Assumptions~\ref{ass:chaine} and \ref{ass:kullback} hold. If $\ty_i$ is generated form component $k_0$ (\emph{i.e.}, $Z_{ik_0}=1$), then for every $k\neq k_0$
\[
\frac{\mathbb{P}(Z_{ik}=1 \mid \ty_i)}{\mathbb{P}(Z_{ik_0}=1 \mid \ty_i)}  \underset{T \to +\infty}{\overset{a.s.}{\longrightarrow}} 0, \quad
\mathbb{P}(Z_{ik_0}=1 \mid \ty_i)  \underset{T \to +\infty}{\overset{a.s.}{\longrightarrow}} 1\; \text{ and }\;
\mathbb{P}(Z_{ik}=1 \mid \ty_i)\underset{T \to +\infty}{\overset{a.s.}{\longrightarrow}} 0.
\]
\end{corollary}
Therefore, by considering $a=1$, Theorem~\ref{thm:conv} and Corollary~\ref{cor:conv} show  that the probability of misclassifying the subject $i$ based on the observation $\ty_i$, using the \emph{maximum a posteriori}  rule, tends to zero when $T$ increases, if the model parameters are known.  
Proof of Theorem~\ref{thm:conv} and a sufficient condition that allows to consider $a=1$ (value of interest when the partition is given by the MAP rule) are given in Appendix, Section~\ref{app:errorclassif}. It should be noted that it is not so common to have an exponential rate of convergence for the ratio of the posterior probability of classification. Similar results are obtained for network clustering using the stochastic block model \citep{CelisseEJS2012} or for co-clustering \citep{Brault2015}. For these two models, the marginal distribution of a single variable provides information about the class membership. For the proposed model, this is the dependency between the different observed variables which is the crucial point for recovering the true class membership. 

\subsection{Dealing with missing values}  \label{sec:NA}
Due to the markovian character of the states, missing values can be handled by iterating the transition matrices.
In our particular context, missing values appear when the accelerometer is not worn (see Section~\ref{sec:data} for explanations of the reasons of missingness). We will not observe isolated missing values but rather wide ranges of missing values. Let $d$ be the number of successive missing values, we thus have to compute the matrix $A_k^{d+1}$ to obtain the distribution of the state at time $t+d$ knowing the state at time $t-1$. These powers of transition matrices should be computed many times during  the algorithm used for inference (see Section~\ref{sec:MLE}). Moreover, after $d+1$ iterations with $d$ large enough, the transition matrix can be considered sufficiently close to stationarity
(\emph{e.g.}, for any $(h,\ell)$, $A_k^{d+1}[h,\ell]\simeq \pi_{k\ell}$), which has actually been chosen as the initial distribution. Therefore, for numerical reasons, we will avoid computing the powers of the transition matrices and we will make the following approximation. An observation $\ty_i$ with $S_i$ observed sequences split with missing value sequences of size at least $d$ are modeled as $S_i$ independent observed sequences with no missing values, all belonging to the same component $k$.
Namely, for each individual $i$, the pdf $p(\ty_{i} ; \bpi_{k}, \bA_k, \blambda, \bvarepsilon)$ of component $k$ is approximated by the product of the pdf of the $S_i$ observed sequences $\ty_{i1}, \ty_{i2}, \ldots, \ty_{i S_i}$:

\begin{equation*}
p(\ty_{i} ; \bpi_{k}, \bA_k, \blambda, \bvarepsilon) \simeq \prod_{s=1}^{S_i} p(\ty_{is} ; \bpi_{k}, \bA_k, \blambda, \bvarepsilon),
\end{equation*}
where, for each $s$, $\ty_{is}$ is an observed sequence of length $T_{is}+1$: $\ty_{is}=(y_{is(0)}, \ldots, y_{is(T_{is})})$ and $p(\ty_{is} ; \bpi_{k}, \bA_k, \blambda, \bvarepsilon)$ is defined as in \eqref{eq:compo}. 
We note that the observation $\ty_i$ can thus be rewritten as follows
\[
\ty_i=(y_{i1(0)}, \ldots, y_{i1(T_{i1})}, y_{i2(0)}, \ldots, y_{i2(T_{i2})}, \ldots, y_{i S_i(0)}, \ldots, y_{i S_i (T_{i S_i})}),
\]
with $y_{i2(0)}= y_{i (T_{i1}+ d_{i1}+1 ) }$ where the $d_{i1}$ values $y_{i (T_{i1}+1)}, \ldots y_{i (T_{i1}+d_{i1})}$ correspond to the first sequence of missing values, and more generally,
for each $s=2, \ldots, S_i$, $y_{is(0)} = y_{i(\sum_{j=1}^{s-1} (T_{ij}+d_{ij}+1))}$, with $d_{ij}$ being the number of missing values between the observed sequences $\ty_{is_j}$ and $\ty_{is_{j+1}}$.

Once the estimation of the parameters has been done, we make sure that this assumption was justified by verifying that the width of the smallest range $d_{min}= \min \left\lbrace d_{i1}, \ldots, d_{i \, S_{i}-1} \right\rbrace$ of missing values is sufficiently large to be greater than the mixing time of the obtained transition matrix. To do so, we use an upper bound for the mixing time given by \citet[Theorem 12.4, p. 155]{Levin2017}.
For each component $k$, we denote by $\nu_k^{*}$ the second maximal absolute eigenvalue of $\bA_k$.
For any positive $\eta$, if for each $k$
\[
d_{min} \geq \frac{1}{1-\nu_k^{*}} \log \frac{1}{\eta \min_{h} \pi_{kh}},
\]
then for any integer $D \geq d_{min}$, the maximum distance in total variation satisfies
$$\max_{h} \|A_k^{D}[h,\cdot]-\pi_{k} \|_{TV} \leq \eta.$$

\section{Maximum likelihood inference}\label{sec:MLE}
This section presents the methodology used to estimate the model parameters.
\subsection{Inference} 
We proposed to estimate the model parameters by maximizing the log-likelihood function where missing values are managed as in Section~\ref{sec:NA} and we recall that the log-likelihood is also approximated for numerical reasons, to avoid computing large powers of the transition matrices. We want to find $\hat\btheta$ which maximizes the following approximated log-likelihood function
$$
\ell_K(\btheta;\ty) = \sum_{i=1}^n  \log \left(\sum_{k=1}^K \delta_k \prod_{s=1}^{S_i} p(\ty_{is} ; \bpi_{k}, \bA_k, \blambda, \bvarepsilon) \right).
$$
This maximization is achieved via an EM algorithm \citep{Dem77} which considers the complete-data log-likelihood defined by
$$
\ell_K(\btheta;\ty,\tz) = \sum_{i=1}^n\sum_{k=1}^K z_{ik} \log \delta_k + \sum_{i=1}^n \sum_{k=1}^K z_{ik}\left( \sum_{s=1}^{S_i}\log p(\ty_{is} ; \bpi_{k}, \bA_k, \blambda, \bvarepsilon) \right).
$$

\subsection{Conditional probabilities}
Let $\alpha_{ikhs(t)}(\btheta)$ be  the probability of the partial sequence $y_{is(0)},\ldots,y_{is(t)}$ and ending up in state $h$ at time $t$ under component $k$. Moreover, let $\beta_{ikhs(t)}(\btheta)$ be the probability of the ending partial sequence $y_{is(t+1)},\ldots,y_{is(T_{is})}$ given a start in state $h$ at time $t$ under component $k$. These probabilities can be easily obtained by the forward/backward algorithm (see Appendix, Section~\ref{sec:conddistribution}). We deduce that the probability $\gamma_{ikhs(t)}(\btheta)$ of being in state $h$ at time $t\in\{0,\ldots,T_{is}\}$ for $\ty_i$ under component $k$ is
\begin{equation*}
\gamma_{ikhs(t)}(\btheta) = \mathbb{P}(X_{is(t)}=h \mid \ty_{is}, Z_{ik}=1;\btheta)= \frac{ \alpha_{ikhs(t)}(\btheta) \beta_{ikhs(t)}(\btheta)}{\sum_{\ell=1}^M  \alpha_{ik\ell s(t)}(\btheta) \beta_{ik\ell s(t)}(\btheta)}.
\end{equation*} 
The probability $\xi_{ikh\ell s(t)}(\btheta)$ of being in state $\ell$ at time $t \in \Omega_i$ and in state $h$ at time $t-1$ for observation $\ty_i$ under component $k$ is
\begin{equation*}
\begin{array}{ll}
\xi_{ikh\ell s(t)}(\btheta) & \vspace{1em} = \mathbb{P}(X_{is(t)}=\ell,X_{is(t-1)}=h \mid \ty_{is}, Z_{ik}=1;\btheta) \\
& = \displaystyle \frac{ \alpha_{ikhs(t)}(\btheta)  \bA_k[h,\ell]  g(y_{is(t)};\blambda_{\ell},\varepsilon_{\ell}) \beta_{ik\ell s(t)}(\btheta)  }{ \sum_{h'=1}^M \sum_{\ell ' =1}^M \alpha_{ikh's(t)}(\btheta)  \bA_k[h',\ell']  g(y_{is(t)};\blambda_{\ell'},\varepsilon_{\ell'}) \beta_{ik\ell' s(t)}(\btheta)}.
\end{array}
\end{equation*}
The probability $\tau_{ik}$ that one observation arises from component $k$ is
\begin{equation*}
\tau_{ik}(\btheta)=\mathbb{P}(Z_{ik}=1\mid\ty_i,\btheta) =  \frac{\prod_{s=1}^{S_i}  \sum_{h=1}^M  \alpha_{ik h s(T_{is})}(\btheta)}{\sum_{k'=1}^K  \prod_{s=1}^{S_i}  \sum_{h=1}^M  \alpha_{ik' h s(T_{is})}(\btheta)}.
\end{equation*}
The probability $\eta_{ihs(t)}$ that observation $i$ is at state $h$ at time $t$ of sequence $s$ is
\begin{equation*}
\eta_{ihs(t)}(\btheta)=\mathbb{P}(X_{is(t)}=h\mid\ty_i,\btheta) = \sum_{k=1}^K \tau_{ik}(\btheta)\gamma_{ikhs(t)}(\btheta).
\end{equation*}

\subsection{EM algorithm}
The EM algorithm is an iterative algorithm randomly initialized at the model parameter $\btheta^{[0]}$. It alternates between two steps: the Expectation step (E-step) consisting in computing the expectation of the complete-data likelihood under the current parameters, and the maximization step (M-step) consisting in maximizing this expectation over the model parameters. Iteration $[r]$ of the algorithm is defined by

\noindent\textbf{E-step} Conditional probability computation, updating of
\[
\tau_{ik}(\btheta^{[r-1]}),\; \gamma_{ikhs(t)}(\btheta^{[r-1]}),\; \eta_{ihs(t)}(\btheta^{[r-1]},\text{ and } \xi_{ikh\ell s(t)}(\btheta^{[r-1]}).
\]
\textbf{M-step} Parameter updating
\[
\delta_k^{[r]} = \frac{n_k(\btheta^{[r-1]})}{n}, \; 
\pi_{kh}^{[r]} = \frac{n_{kh(0)}(\btheta^{[r-1]})}{n_k(\btheta^{[r-1]})}, \;
\bA_{k}[h,\ell]^{[r]} = \frac{n_{kh\ell}(\btheta^{[r-1]})}{n_{kh}(\btheta^{[r-1]})}, \;
\varepsilon_h^{[r]} = \frac{w_{h}(\btheta^{[r-1]})}{n_{kh}(\btheta^{[r-1]})},
\]
\[\text{and } \;
\blambda_h^{[r]} = \argmax_{\lambda_h}\sum_{i=1}^n \sum_{s = 1}^{S_i} \sum_{t=0}^{T_{is}} \eta_{ihs(t)}(\btheta^{[r-1]}) g_c(y_{is(t)};\blambda_h),
\]
where 
\begin{align*}
n_k(\btheta)=&\sum_{i=1}^n \tau_{ik}(\btheta),~
n_{kh}(\btheta) = \sum_{i=1}^n\sum_{s=1}^{S_i} \sum_{t=0}^{T_{is}}\tau_{ik}(\btheta) \gamma_{ikhs(t)},~  n_{kh(0)}(\btheta)=\sum_{i=1}^n\sum_{s=1}^{S_i} \tau_{ik}(\btheta) \gamma_{ikhs(0)}(\btheta),\\
n_{kh\ell}(\btheta) &=\sum_{i=1}^n \sum_{s = 1}^{S_i} \sum_{t=1}^{T_{is}} \tau_{ik}(\btheta)\xi_{ikh\ell s(t)}(\btheta)\; \text{ and } w_{h}(\btheta)  = \sum_{i=1}^n \sum_{s = 1}^{S_i} \sum_{t=0}^{T_{is}} \eta_{ihs(t)}(\btheta) \mathbf{1}_{\{y_{is(t)}=0\}}.
\end{align*}

\section{Numerical illustrations}\label{sec:simu}
This section aims to highlight the main properties of the model on numerical experiments.
First, simulated data are used to illustrate the exponential decay of the probabilities of misclassification (given by Theorem~\ref{thm:conv}), the convergence of estimators and the robustness of the approach to missingness. Second, our approach is applied to the data from the PAT study. All the experiments are conducted with the \proglang{R} package \pkg{MHMM} available on \proglang{CRAN}.

\subsection{Simulated data}
\paragraph{Simulation design}
All the simulations are performed according to the same model. 
This model is a bi-components mixture of HMM with two states (\emph{i.e.}, $K=M=2$) and equal proportions (\emph{i.e.}, $\delta_1=\delta_2=1/2$). The distribution of $Y_{i(t)}$ conditionally on the state $h$ is a ZIG distribution. We have
\[
\varepsilon_1=\varepsilon_2 = 0.1,\;
a_1 = 1,\; b_1=b_2=1,\; \bA_1 = \begin{bmatrix}
e & 1-e \\
1-e & e \\
\end{bmatrix}
\text{ and } 
\bA_2 = \begin{bmatrix}
1-e & e \\
e & 1-e \\
\end{bmatrix}.
\]
The parameter $a_2>1$ controls the separation of the distribution of $Y_{i(t)}$ given the state. The parameter $e$ controls the separation of the distribution of $X$ given the class (when $e$ increases, the constant $c$ in Theorem~\ref{thm:conv} increases).
We consider four cases: hard ($e=0.75$ and $a_2=3$), medium-hard ($e=0.90$ and $a_2=3$), medium-easy ($e=0.75$ and $a_2=5$) and easy ($e=0.90$ and $a_2=5$). 

\paragraph{Illustrating the exponential rate of the probabilities of misclassification}
Theorem~\ref{thm:conv} states that the probabilities of misclassification decrease at an exponential rate with $T$. To illustrate this property, 1000 sequences are generated for $T=1,\ldots,100$ and the four cases. For each sequence $\ty_i$, we compute $\log({\mathbb{P}(Z_{ik}=1 \mid \ty_i)}/{\mathbb{P}(Z_{ik_0}=1 \mid \ty_i)})$ when $k_0$ is the true class, $k$ the alternative and the true model parameters are used.
Figure~\ref{fig:simu1}(a) shows the behavior of $\log({\mathbb{P}(Z_{ik}=1 \mid \ty_i)}/{\mathbb{P}(Z_{ik_0}=1 \mid \ty_i)})$ (the median of this log ratio is plotted in plain and a 90$\%$ confidence interval is plotted in gray). Note that this log ratio of probabilities linearly decreases with $T$ which illustrates the exponential decay of the probabilities of misclassification. Moreover, Figure~\ref{fig:simu1}(b) presents the empirical probabilities of misclassification and thus also illustrates Theorem~\ref{thm:conv}. As expected, this shows that the decay of the probabilities of misclassification is faster as the overlaps between class decreases.

\begin{figure}[t]
\begin{center}
\subfigure[Median (in plain) and 90\%-confidence region (gray area) of $\log\frac{\mathbb{P}(Z_{ik}=1 \mid \ty_i)}{\mathbb{P}(Z_{ik_0}=1 \mid \ty_i)}$.]{
\includegraphics[width=0.45\textwidth]{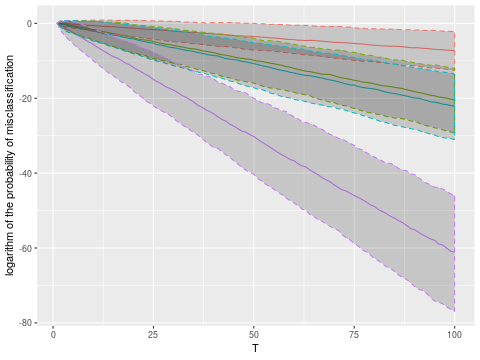}}
\subfigure[Probability of misclassification.]{
\includegraphics[width=0.45\textwidth]{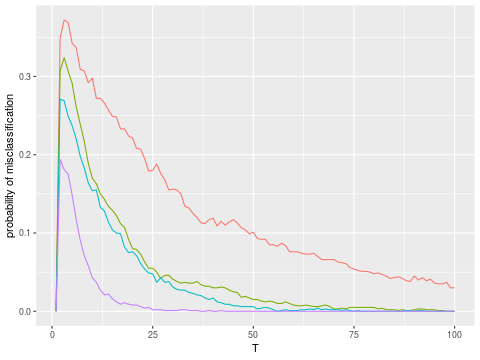}}
\end{center}
\caption{Results obtained on 1000 observations for the hard (orange), medium-hard (green), medium-easy (blue) and easy (purple) cases.}
\label{fig:simu1}
\end{figure}

\paragraph{Illustrating the convergence of the estimators}
We illustrate the convergence of the estimators (partition, latent states and parameters) when the model parameters are estimated by maximum likelihood (see Section~\ref{sec:MLE}). We compute the mean square error (MSE) between the model parameters and their estimators. Moreover, we compute the adjusted Rand index (ARI; \citet{Hub85}) between the true partition and the partition given by the MAP rule, and between the true state sequences and the estimated state sequences given by the MAP rule (obtained with the Viterbi algorithm \citep{viterbi1967error}). Table~\ref{tab:resSimMLE} shows the results obtained with two different sample sizes $n$ and two different lengths of sequences $T$, considering the case medium-hard. It can be seen that the partition and the model parameters are well estimated. Indeed, the MLE converge to the true parameters as $T$ or $n$ increases, except for the proportion of each component $\delta_k$. The convergence of the estimator of the proportions depends mainly on the sample size $n$. 
We notice that the partition obtained by our estimation procedure corresponds to the true partition (for $n$ and $T$ large enough) even if we are not under the true parameters but under the MLE, which is not an immediate consequence of Theorem~\ref{thm:conv}. On the contrary, we do not find the true state sequences a.s., as the number of states to be estimated is also growing with $n$ and $T$. This result was expected because the number of latent states increases with $T$ and $n$ while the number of parameters and the dimension of the partition does not increase with $T$. Results obtained for the three other cases are similar and are presented in Appendix, Section~\ref{app:simu}. 

\begin{table}[h!t]
\caption{Convergence of estimators when 1000 replicates are drawn from case medium: ARI between estimated and true partition, ARI between estimated and true latent states and MSE between the MLE and the true parameters}
\label{tab:resSimMLE}
\centering
\begin{tabular}{llccccccc}
\hline
 & & \multicolumn{2}{c}{ARI (latent variables)} & \multicolumn{5}{c}{MSE (model parameters)} \\
$n$ & $T$ & partition & states & $\bA_k$ & $\varepsilon_h$ & $a_h$ & $b_h$ & $\delta_k$ \\
\hline
10 & 100 &  0.995 & 0.621 & 0.021 & 0.001 & 0.088 & 0.024 & 0.047\\
10 & 500 & 1.000 & 0.632 & 0.007 & 0.000 & 0.020 & 0.005 & 0.048 \\
100 & 100 & 0.996 & 0.630 & 0.004 & 0.000 & 0.011 & 0.003 & 0.005\\
100 & 500 &  1.000 & 0.634 & 0.003 & 0.000 & 0.005 & 0.002 & 0.005 \\
\hline
\end{tabular}
\end{table}

\paragraph{Illustrating the robustness to missingness}
We now investigate the robustness of the proposed method with missingness. We compare the accuracy of the estimators (ARI for the latent variables and MSE for the parameters) obtained on samples without missingness to the accuracy of the estimators obtained when missingness is added to the samples. Three situations of missingness are considered: missing completely at random-1 (MCAR-1) (\emph{i.e.}, one sequence of 10 missing values is added to each sequence $\ty_i$, the location of the sequence follows a uniform distribution), MCAR-2 (\emph{i.e.}, two sequences of 20 missing values are added for each sequence $\ty_i$, the location of the sequences follows a uniform distribution) and missing not at random (MNAR) (\emph{i.e.}, the probability to observe the value $y_{i(t)}$ is equal to ${e^{y_{i(t)}}}/{(1+e^{y_{i(t)}})}$). Note that the last situation adds many missing values when the true value of $y_{i(t)}$ is close to zero, so the occurrence of missing values depends on the latent states.
Table~\ref{tab:resSimNA} compares the results obtained with and without missingness, considering case medium-hard. It shows that estimators are robust to missingness. Results obtained  for the other three cases are similar and are reported in Appendix, Section~\ref{app:simu}.

\begin{table}[h!t]
\caption{Convergence of estimators obtained over 1000 replicates with and without missing data when data are sampled from case medium: ARI between estimated and true partition, ARI between estimated and true latent states and MSE between the MLE and the true parameters}
\label{tab:resSimNA}
\centering
\begin{tabular}{lllccccccc}
\hline
 & &  & \multicolumn{2}{c}{Adjusted Rand index} & \multicolumn{5}{c}{Mean square error} \\
$n$ & $T$ & missingness & partition & states & $\bA_k$ & $\varepsilon_h$ & $a_h$ & $b_h$ & $\delta_k$ \\
\hline
   10 &100 & no missingness & 0.995 & 0.621 & 0.021 & 0.001 & 0.088 & 0.024 & 0.047 \\ 
 && MCAR-1 & 0.991 & 0.613 & 0.024 & 0.001 & 0.102 & 0.028 & 0.047 \\ 
 && MCAR-2 & 0.987 & 0.605 & 0.028 & 0.001 & 0.113 & 0.032 & 0.047 \\ 
 && MNAR & 0.934 & 0.497 & 0.051 & 0.003 & 0.398 & 0.050 & 0.050 \\ 
10&500& no missingness & 1.000 & 0.632 & 0.007 & 0.000 & 0.020 & 0.005 & 0.048 \\ 
 && MCAR-1 & 1.000 & 0.631 & 0.007 & 0.000 & 0.020 & 0.005 & 0.048 \\ 
 && MCAR-2 & 1.000 & 0.631 & 0.007 & 0.000 & 0.019 & 0.005 & 0.048 \\ 
 && MNAR & 0.999 & 0.516 & 0.021 & 0.003 & 0.233 & 0.028 & 0.048 \\ 
100&100& no missingness & 0.996 & 0.630 & 0.004 & 0.000 & 0.011 & 0.003 & 0.005 \\ 
 && MCAR-1 & 0.994 & 0.624 & 0.004 & 0.000 & 0.013 & 0.003 & 0.005 \\ 
 && MCAR-2 & 0.989 & 0.618 & 0.005 & 0.000 & 0.014 & 0.004 & 0.005 \\ 
 && MNAR & 0.951 & 0.512 & 0.014 & 0.002 & 0.200 & 0.026 & 0.005 \\ 
 100&500& no missingness & 1.000 & 0.634 & 0.003 & 0.000 & 0.005 & 0.002 & 0.005 \\ 
 && MCAR-1 & 1.000 & 0.633 & 0.002 & 0.000 & 0.006 & 0.002 & 0.005 \\ 
 && MCAR-2 & 1.000 & 0.632 & 0.002 & 0.000 & 0.005 & 0.002 & 0.005 \\ 
 && MNAR & 1.000 & 0.520 & 0.011 & 0.002 & 0.198 & 0.026 & 0.005 \\ 
\hline
\end{tabular}
\end{table}
    
\subsection{Using the approach on classical accelerometer data}
We consider the accelerometer data measured on three subjects available from \citet{huang2018hidden}. The accelerometer measures the activity every five minutes for one week. 
Note that the first subject has $2\%$ of missing values. The purpose of this section is to illustrate the differences between the method of \citet{huang2018hidden} and the method proposed in this paper.

\citet{huang2018hidden} consider one HMM per subject with three latent states. This model is used for monitoring the circadian rhythmicity, subject by subject. Because they fit one HMM per sequence measured by the accelerometer of a subject, the definition of the activity level is different for each subject (see, \citet[Figure~4]{huang2018hidden}). This is not an issue for their study because the analysis is done subject by subject. However, the mean time spent by activity levels cannot be compared among the subjects.
\begin{figure}[ht!]
\centering 
\includegraphics[width=0.95\textwidth]{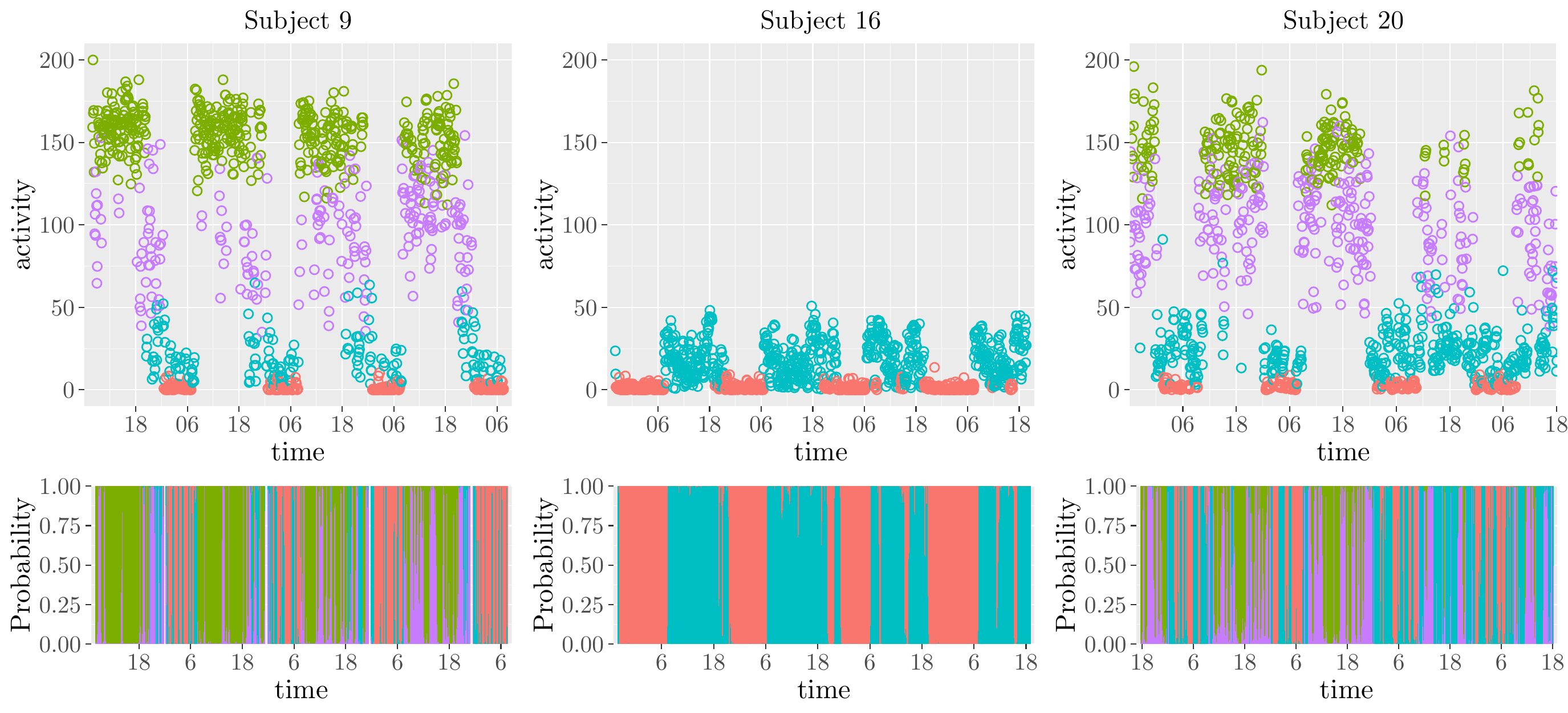}
\includegraphics[width=0.4\textwidth]{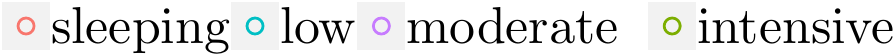}
\caption{State estimation for the three subjects: (top) accelerometer data where color indicates the expected value of $Y_{i(t)}$ conditionally to the most likely state and to the most likely component; (bottom) probability of each state at each time.}
\label{fig::huangdata}
\end{figure}
The method proposed here makes this comparison possible. Figure~\ref{fig::huangdata} depicts the activity data of the three subjects, the expected value of $Y_{i(t)}$ conditionally to the most likely state and on the most likely component and the probability of each state. Based on the QQ-plot (see, Appendix, Section~\ref{app:huang}), we consider $M=4$ activity levels. These levels can be easily characterized with the model parameters presented in Table~\ref{tab:paramHuang}. Moreover, the transition matrices also make sense. For instance, class~1 (subjects 9 and 20) has an almost tri-diagonal transition matrix (by considering an order between the states given through the activity levels per state) and class-2 (subject 2) is composed of a subject with low-overall activity
\[
\hat{\bA}_1 = 
\begin{bmatrix}
0.86 & 0.14 & 0.00 & 0.00\\
0.12 & 0.81 & 0.06 & 0.01\\
0.00 & 0.07 & 0.79 & 0.14\\
0.00 & 0.00 & 0.13 & 0.87\\
\end{bmatrix}.
\]

\begin{table}[h!t]
\caption{Parameters  and mean time per states for the three subjects}
\label{tab:paramHuang}
\centering
\begin{tabular}{lccccc}
\hline
State name & $\varepsilon_h$ & $a_h$ & $b_h$ & mean & sd  \\ 
\hline
intensive-level& 0.00 & 98.94 & 0.65 & 152.76 & 15.36\\
moderate-level& 0.00 & 11.09 & 0.11 & 99.34 & 29.84\\
low-level& 0.00 & 2.32 & 0.11 &   20.98 & 13.79\\
sleeping   & 0.22 & 1.48 & 0.72 &  2.06 & 1.70\\
\hline
\end{tabular}
\end{table}

\section{Analysis of PAT data} \label{sec:appli}
In this section, we analyze the data presented in Section~\ref{sec:data}.

\subsection{Experimental conditions}
In order to compare our approach to the cuts defined \emph{a priori} in the PAT study (see Section~\ref{sec:data}), the model was fitted with four activity levels. Note that selecting the number of states in HMM stays a challenging problem (see the discussion in the conclusion). However, approaches considering four activity levels are standard for accelerometer data. The number of components (\emph{i.e.}, the number of classes) is estimated, using an information criterion unlike the PAT study where it is arbitrarily set at 3 or 4. For each number of components, 5000 random initializations of the EM algorithm are performed. The analysis needs about one day of computation on a 32-Intel(R) Xeon(R) CPU E5-4627 v4 @ 2.60GHz. 
 
\subsection{Model selection}
To select the number of components, we use two information criteria which are generally used in clustering: the BIC \citep{Schwarz:78} and the ICL \citep{Biernacki:00} defined by
\[
\mathrm{BIC}(K) = \ell_K(\btheta;\ty) - \frac{\nu_K}{2} \log (\sum_{i=1}^n \sum_{s=1}^{S_i} T_{is}+1),
\]
and 
\[
\mathrm{ICL}(K) = \mathrm{BIC}(K) + \sum_{i=1}^n \sum_{k=1}^K \hat z_{ik}(\hat\btheta) \log \tau_{ik}(\hat\btheta),
\]
where $\nu_K= (K-1) + K(M + M^2) + 3M$ is the number of parameters for a model with $K$ components and $M$ states and $\hat z_{ik}(\hat\btheta)$ defines the partition by the MAP rule associated to the MLE such that
\[
\hat z_{ik}(\hat\btheta) = \left\{\begin{array}{rl}
1 & \text{if } \tau_{ik}(\hat\btheta)=\mathrm{argmax}_{\ell=1,\ldots,K}\;  \tau_{i\ell}(\hat\btheta)\\
0 & \text{otherwise}
\end{array}\right. .
\]
The ICL is defined according as the integrated complete-data likelihood computed with the partition given by the MAP rule with the MLE. The values of the information criteria are given in Table~\ref{tab:criteria}, for different number of classes. Both criteria select five components. The values of $\mathrm{ICL}(K)$ are close to those of the $\mathrm{BIC}(K)$,   implying that the entropy $\sum_{i=1}^n \sum_{k=1}^K \hat z_{ik} \log \tau_{ik}(\hat\btheta)\approx 0$. This is a consequence of Theorem~\ref{thm:conv} (see also numerical experiments in Section~\ref{sec:simu}). In the following, we interpret the results obtained with $M=4$ activity levels and $K=5$ classes.
 \begin{table}[h!t]
   \caption{Information criteria obtained on PAT data with four levels of activity (minima are in bold)}
\label{tab:criteria}
 \centering
\begin{tabular}{lccccccc}
\hline
$K$ & 1 & 2 & 3 & 4 & 5 & 6 &  7 \\ 
\hline
BIC   &  -2953933 & -2952313 & -2951809 & -2951705 & \textbf{-2951308} & -2951364 & -2951696\\
ICL   &  -2953933 & -2952313 & -2951810 & -2951707 & \textbf{-2951309} & -2951364 & -2951697\\
\hline
\end{tabular}
\end{table}

\subsection{Description of the activity levels}\label{subsec:level}
The parameters of the ZIG distributions are presented in Table~\ref{tab:paramlevels}. The four distributions are ordered by the value of their means. The \emph{sleeping state} is characterized by a large probability of observing zero (\emph{i.e.}, $\varepsilon_h$ is close to one). However, $\varepsilon_h$ is not equal to zero for the other states, but the more active the state is, the smaller $\varepsilon_h$ is. We also compute the marginal cutoffs (\emph{i.e.}, the cutoffs by considering the MAP of $\mathbb{P}(X_{i(t)}\mid Y_{i(t)})$). These cutoffs neglect the time dependency due to the Markov structure, but can be compared to the cutoffs proposed by the PAT study. Indeed, according to the PAT study, minutes with $<100$ counts are assigned to Sedentary activity, minutes with 100-2019 counts were classified as Light, the class Moderate corresponds to 2020-5998 counts/minute and Vigorous 5999 and above counts/minute.
The marginal cutoff associated with the \emph{low-level} state is very close to that of the Sedentary class of the PAT. We find, however, that our marginal cutoffs are more accurate for higher levels of activity. PAT cutoffs do not adequately characterize the activity level of the study population. Finally, contrary to classical thresholds, our modeling approach allows to capture and characterize the variability associated with the different levels of activity, variability which seems important (see Figure~\ref{fig:observations} and Table~\ref{tab:paramlevels}).

\begin{table}[h!t]
\caption{Parameters describing the four activity levels for PAT data and statistics of the distributions}
\label{tab:paramlevels}
\centering
\begin{tabular}{lccccc}
\hline
 Name of the activity level & \multicolumn{3}{c}{Parameters} & \multicolumn{2}{c}{Statistics}\\ 
 & $\varepsilon_h$ & $a_h$ & $b_h$ & mean & marginal cutoffs\\ 
 \hline
sleeping  & 0.988 & 7.470 & 7.470 &   0.012 & $[0,\;0]$\\
low-level& 0.260 &  0.974 & 0.020 &   36.926 & $]0,\;97.7]$\\
moderate-level& 0.025  & 1.408  & 0.004 &  329.249& $]97.7,\; 614.4]$\\
intensive-level& 0.007 & 2.672 & 0.002 & 1696.935& $]614.4,\;+\infty[$\\
\hline
\end{tabular}
\end{table}

\subsection{Description of the classes}
Classes can be described using their proportions and their associated parameters presented in Table~\ref{tab:paramlevels}. 
The data are composed of a majority class ($\delta_1=0.518$). Three other classes are composed of more sedentary individuals (\emph{e.g.}, their marginal probabilities of being in states 1 and 2 are higher). Finally, there is a small class ($\delta_5=0.045$) which contains the most active subjects (\emph{i.e.}, $\pi_{k4}=0.143$). 
For three of the five classes,  Figure~\ref{fig:observations} presents a characteristic subject of each class and the probabilities of the activity levels (the associated graphs for the two remaining classes are given in Appendix, Section~\ref{supp:pat}).
Classes can be interpreted from the mean time spent at different activity levels presented in Table~\ref{tab:meantime} and from transition matrices presented in Table~\ref{tab:usual} which are almost tri-diagonal. This could be expected because it seems relevant to obtain a low probability of jumping between the sleeping state and the intensive state. Additionally, the approximation made for efficiently handling the missingness (see Section~\ref{sec:NA}) turns out to be relevant. The minimal range of missing values is indeed equal to $d_{\min}=60$ which leads to a distance in total variation between the $d_{\min}$-power of the transition matrices and the stationary distribution being less than $5.10^{-4}$ for any component.

\begin{figure}[H]
\begin{center}
\subfigure[\emph{Moderate} class]{\includegraphics[width=0.4\textwidth]{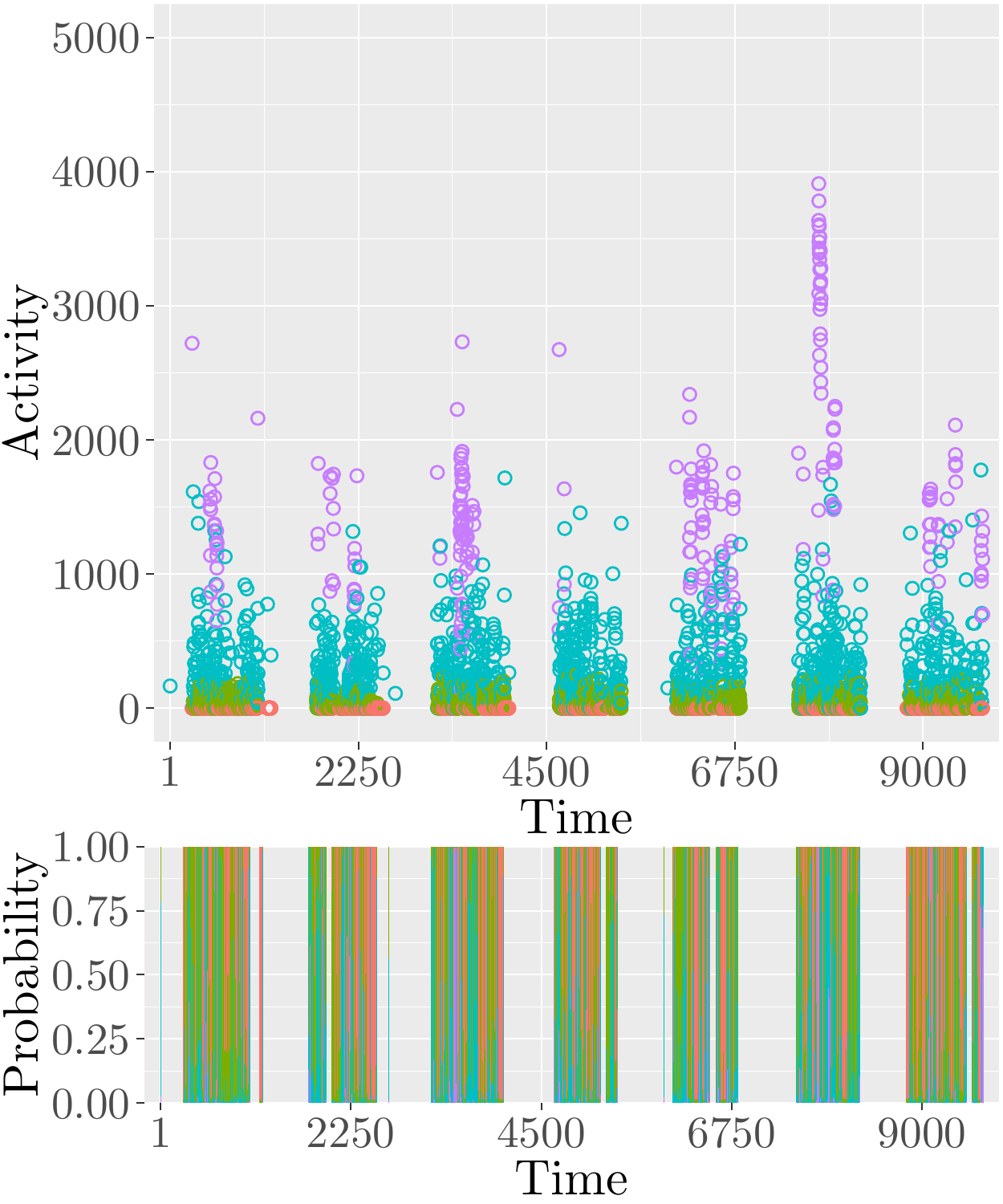}}
\subfigure[\emph{Very sedentary} class]{\includegraphics[width=0.4\textwidth]{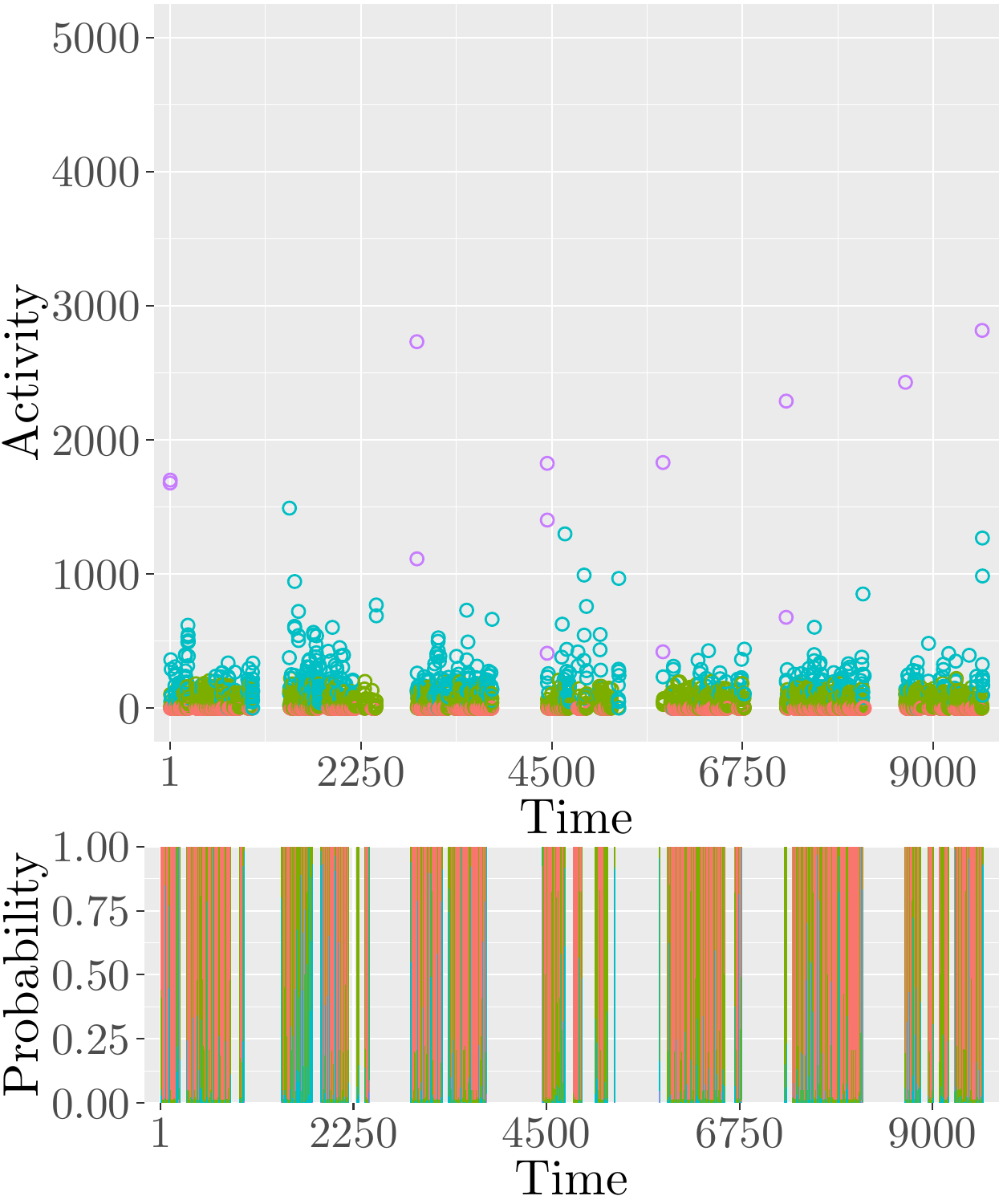}}
\subfigure[\emph{Very active} class]{\includegraphics[width=0.4\textwidth]{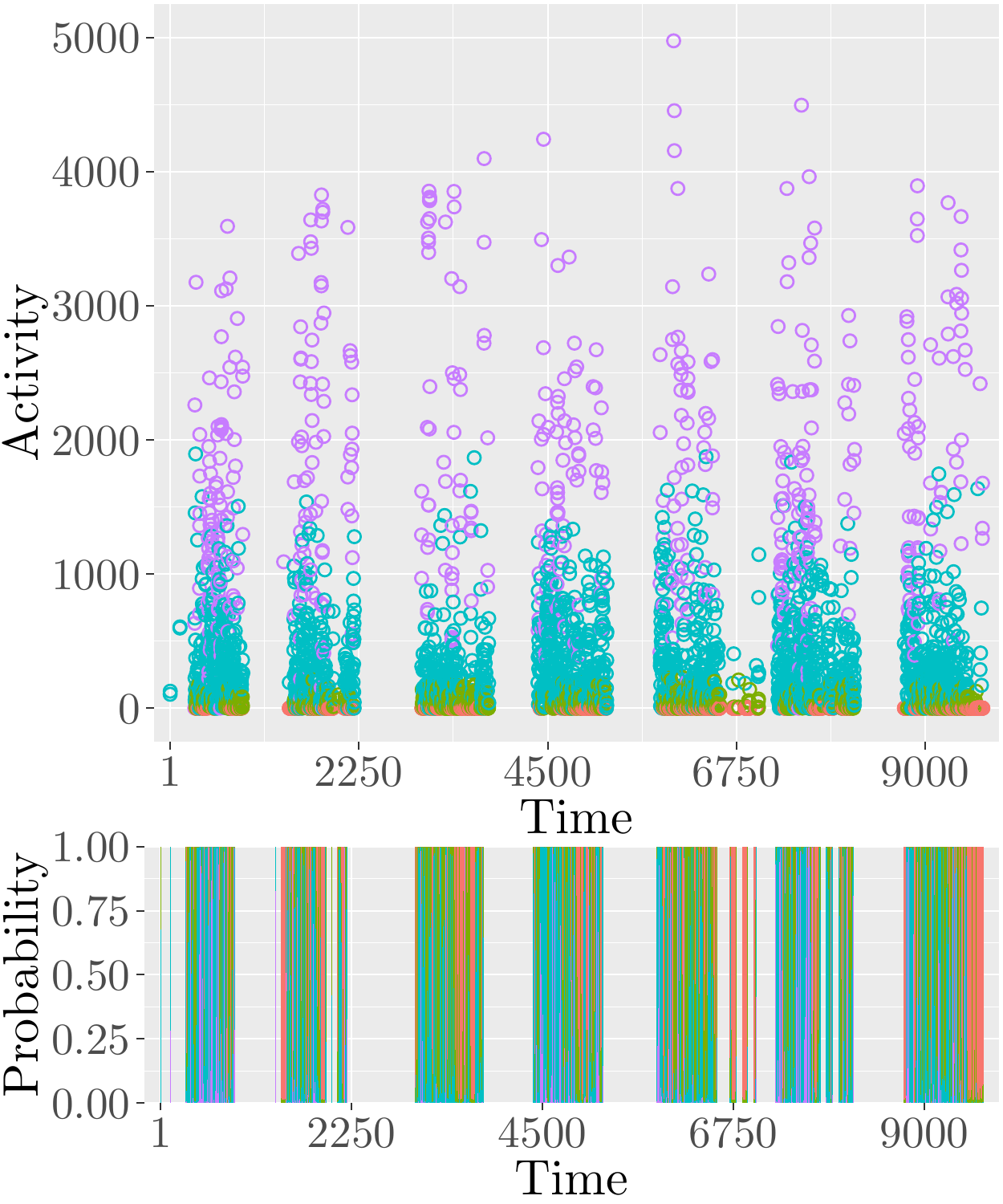}}
\end{center}
\begin{center}
\includegraphics[width=0.4\textwidth]{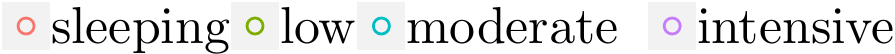}
\end{center}
\caption{Examples of observations assigned into the five classes with the probabilities of the states.}
\label{fig:observations}
\end{figure}

\begin{table}[ht]
\caption{Mean time spent at the different activity levels for the five classes}
\label{tab:meantime}
\centering
\begin{tabular}{lcccc}
\hline
& sleeping & low-level & moderate-level & intensive-level\\
\hline
Class \emph{active}  & 0.306 & 0.284 & 0.338 & 0.072 \\
Class \emph{sedentary}  & 0.467 & 0.209 & 0.263 & 0.061 \\
Class \emph{moderate}  & 0.304 & 0.411 & 0.225 & 0.060 \\
Class \emph{very sedentary}  & 0.504 & 0.366 & 0.124 & 0.006 \\
Class \emph{very active}  & 0.189 & 0.351 & 0.316 & 0.143\\
\hline
\end{tabular}
\end{table}

\begin{table}[ht]
\caption{Transition matrix for the five classes}
\label{tab:usual}
\centering
\begin{tabular}{lcccc}
\hline
&\multicolumn{4}{c}{Class \emph{moderate}}\\
 & sleeping & low-level & moderate-level & intensive-level \\   
 \hline
sleeping & 0.76 & 0.21 & 0.03 & 0.00 \\ 
low-level & 0.16 & 0.73 & 0.11 & 0.00 \\ 
moderate-level & 0.03 & 0.20 & 0.73 & 0.04 \\ 
intensive-level & 0.01 & 0.04 & 0.16 & 0.80 \\ 
\hline
&\multicolumn{4}{c}{Class \emph{very sedentary}}\\
 & sleeping & low-level & moderate-level & intensive-level \\   
 \hline
sleeping & 0.85 & 0.08 & 0.06 & 0.00 \\ 
low-level & 0.20 & 0.67 & 0.13 & 0.01 \\ 
moderate-level & 0.10 & 0.11 & 0.76 & 0.03 \\ 
intensive-level & 0.01 & 0.04 & 0.14 & 0.82 \\ 
\hline
&\multicolumn{4}{c}{Class \emph{very active}}\\
 & sleeping & low-level & moderate-level & intensive-level \\ 
 \hline
sleeping & 0.80 & 0.14 & 0.05 & 0.01 \\ 
low-level & 0.08 & 0.74 & 0.17 & 0.01 \\ 
moderate-level & 0.03 & 0.18 & 0.69 & 0.10 \\ 
intensive-level & 0.01 & 0.05 & 0.21 & 0.74 \\
\hline
\end{tabular}
\end{table}

\section{Conclusion}\label{sec:concl}
A specific mixture of HMM has been introduced to analyze accelerometer data. It avoids the traditional cutoff point method and provides a better characterization of activity levels for the analysis of these data, while adapting to the population. The proposed model could be applied to a population with different characteristics (\emph{e.g.}, younger) which would lead to different definitions of activity levels. In addition, the use of several HMMs  make to take into account dependency over time and thus improve the traditional method based on cutoff points \citep{witowski2014using}. This approach also allows us to take into account the heterogeneity of the population (in the sense of physical activity).

An interesting perspective is to consider adjusting for confusing factors  (\emph{e.g.}, gender or age). These confusing factors could impact the probabilities of transition between the latent spaces (\emph{e.g.}, using a generalized linear model approach) and/or the definition of the accelerometer measurement given a state (\emph{e.g.}, linear regression on some parameters of the ZIG distribution).

In the application, the number of activity levels was not estimated but fixed at a common value for accelerometer data. 
Estimating the number of states for a mixture of HMM is an interesting but complex topic: for instance, the use of BIC is criticized (see, \citet[Chapter 15]{Cappe}). 
This makes the study of relevant information criteria for selecting the number of states an interesting topic  for future work. Pseudo-likelihood based criteria could be used \citep{gassiat2002likelihood,csiszar2006} but the fact that the marginal distribution of one $Y_{i(t)}$ is not identifiable limits this approach. A more promising approach could be to use cross-validated likelihood \citep{CeleuxCS2008} but it would be computationally intensive because accelerometer data provides a large number of observations.

\appendix

\section{Model identifiability}\label{app:ident}
The proof of Theorem~\ref{thm:ident} is split in two parts: 
\begin{enumerate}
\item Identifiability of the parameters of the specific distribution per state is obtained using the approach of \citet{teicher1963}. Hence $\forall h=1,\ldots,M$
\[
\blambda_h=\tilde{\blambda}_h,\quad\sum_{k=1}^K\delta_k\pi_{kh}(1-\varepsilon_h)=\sum_{k=1}^K\tilde \delta_k \tilde \pi_{kh}(1-\tilde \varepsilon_h).
\]
\item Identifiability of the transition matrices and of the $\varepsilon$ is shown using properties of Vandermonde matrices. Hence, 
\[
\forall k=1,\ldots,K,\;\delta_k=\tilde\delta_k,\; \bA_k=\tilde{\bA}_k,\; \pi_k=\tilde\pi_k,   \bvarepsilon = \tilde \bvarepsilon.
\]
\end{enumerate}

\subsection{Identifiability of the parameters of the specific distribution per state}
Considering the marginal distribution at time $t=0$, we have
\[
\sum_{k=1}^K\sum_{h=1}^M \delta_k \pi_{kh} g(y_{i(0)};\blambda_h, \varepsilon_h) = \sum_{k=1}^K \sum_{h=1}^M\tilde\delta_k \tilde\pi_{kh} g(y_{i(0)};\tilde\blambda_h, \tilde\varepsilon_h).
\]
Note that $ g(y_{i(0)};\blambda_h, \varepsilon_h) = (1-\varepsilon_h) g_c(y_{i(0)};\blambda_h) + \varepsilon_h \mathbf{1}_{\{y_{i(0)}=0\}} $ is a pdf of a zero-inflated distribution, so it is a pdf of a bi-component mixture. We now use the same reasoning as \citet{teicher1963}. We have
\[
1 + \sum_{h=2}^M \frac{ g(y_{i(0)};\blambda_h, \varepsilon_h) \displaystyle\sum_{k=1}^K \delta_k \pi_{kh} }{g(y_{i(0)};\blambda_1, \varepsilon_1)\displaystyle \sum_{k=1}^K \delta_k \pi_{k1}} = \frac{ g(y_{i(0)};\tilde\blambda_1, \tilde\varepsilon_1) \displaystyle\sum_{k=1}^K \tilde\delta_k \tilde\pi_{k1} }{g(y_{i(0)};\blambda_1, \varepsilon_1) \displaystyle\sum_{k=1}^K \delta_k \pi_{k1}} + \sum_{h=2}^M \frac{ g(y_{i(0)};\tilde\blambda_h, \tilde\varepsilon_h) \displaystyle\sum_{k=1}^K \tilde\delta_k \tilde\pi_{kh} }{g(y_{i(0)};\blambda_1, \varepsilon_1) \displaystyle\sum_{k=1}^K \delta_k \pi_{k1}}.
\]
Considering $y_{i(0)} \to \rho$, by Assumption \ref{ass:order}, we have
\[
\blambda_1 = \tilde \blambda_1, \quad (1-\varepsilon_1) \sum_{k=1}^K \delta_k \pi_{k1} = (1-\tilde\varepsilon_1) \sum_{k=1}^K \tilde\delta_k \tilde \pi_{k1}.
\]
Repeating the previous argument with $h=2,\ldots,M$, we conclude that, for $h\in\{1,\ldots,M\}$,
\begin{equation*}
\blambda_h = \tilde \blambda_h, \quad (1-\varepsilon_h) \sum_{k=1}^K \delta_k \pi_{kh} = (1-\tilde\varepsilon_h) \sum_{k=1}^K \tilde\delta_k \tilde \pi_{kh}.
\end{equation*}

\subsection{Identifiability of the transition matrices}
First, we introduce two technical lemmas of which proofs are discussed in the next subsection. Second, we show that $\bA_k[1,1]=\tilde \bA_k[1,1]$ then we extend the results to the whole transition matrices.

\begin{lemma}\label{lem:tech1}
Let $N_0$, $N_1$, $\tilde N_0$ and $\tilde N_1$ be four definite positive matrices of size $K\times K$ such that for $u\in\{1,\ldots,K\}$ and $k\in\{1,\ldots,K\}$,
\begin{equation*}
N_0[u,k] = a_k^{u-1}, \; N_1[u,k]=a_k^{K + u -1}, \; \tilde N_0[u,k] = \tilde a_k^{u-1}, \; \tilde N_1[u,k]= \tilde a_k^{K + u -1}, \;
\end{equation*}
with $a_k > a_{k+1} > 0$, $\tilde a_k > \tilde a_{k+1} > 0$ and $a_1\geq \tilde a_1$. If for any $\tilde w \in \mathbb{R}^K_+$ there exists $w \in \mathbb{R}^K_+$ $N_0  w = \tilde N_0 \tilde w$ and $N_1 w = \tilde N_1 \tilde w$ then for $k\in\{1,\ldots,K\}$ $a_k=\tilde a_k$ and $w=\tilde w$.
\end{lemma}

\begin{lemma}\label{lem:tech2}
Let $N_0$,  $\tilde N_0$   be two definite positive matrices of size $K\times K$ such that for $u\in\{1,\ldots,K\}$ and $k\in\{1,\ldots,K\}$,
\begin{equation*}
N_0[u,k] = a_k^{u-1}, \; N_1[u,k]=a_k^{K + u -1},
\end{equation*}
with $a_k > a_{k+1} > 0$, $\tilde a_k > \tilde a_{k+1} > 0$ and $a_1\geq \tilde a_1$. 
Let $D_u =\mathrm{diag}(a_1^{Ku},\ldots,a_K^{Ku})$ and $\tilde D_u =\mathrm{diag}(\tilde a_1^{Ku},\ldots, \tilde a_K^{Ku})$. If there exist $\alpha\in]0,1[$, $\tilde\alpha\in]0,1[$, $w\in \mathbb{R}^K_+$ and $\tilde w\in \mathbb{R}^K_+$ such that for $u\in\{0,\ldots,K-1\}$, we have
\begin{equation*}
\alpha N_0 D_u w = \tilde \alpha \tilde N_0 \tilde D_u \tilde w,
\end{equation*}
then for $k\in\{1,\ldots,K\}$ $a_k=\tilde a_k$ and $w=\tilde w$.
\end{lemma}

We consider the marginal distribution of $(y_{i(0)},\ldots,y_{i(t-1)})$ with $t=1,\ldots,2K$, where $y_{i(0)}=y_{i(t')}$ for each $t'=1,\ldots,t-3$, $y_{i(t-2)} = y_{i(0)}^{\tau_1}$ , $y_{i(t-1)} = y_{i(0)}^{\tau_2}$ and $y_{i(t)} = y_{i(0)}^{\tau_3}$. Therefore, taking $\tau_1=\tau_2=\tau_3=1$ and letting $y_{i(0)}$ tend to $\rho$ (see Assumption~\ref{app:ident}), we obtain, for $t=1,\ldots,2K$, that
\begin{equation*}
(1-\varepsilon_1) \sum_{k=1}^K\delta_k \pi_{k1}  \left(\bA_k[1,1](1-\varepsilon_1)\right)^{t-1} = (1-\tilde\varepsilon_1) \sum_{k=1}^K\tilde \delta_k \tilde\pi_{k1}  \left(\tilde \bA_k[1,1](1-\tilde \varepsilon_1)\right)^{t-1}.
\end{equation*}
Because, we consider $2K$ marginal distributions, we can use Lemma~\ref{lem:tech2} by setting $\alpha=(1-\varepsilon)$, $\tilde\alpha=(1-\tilde\varepsilon)$, $a_k=\bA_k[1,1](1-\varepsilon_1)$, $\tilde a_k = \tilde \bA_k[1,1](1-\tilde \varepsilon_1)$, $w_k = \delta_k \pi_{k1}$ and $w_k = \tilde\delta_k \tilde\pi_{k1}$. Therefore, we have $\varepsilon=\tilde\varepsilon$, $\bA_k[1,1]=\tilde \bA_k[1,1]$ and $\delta_k \pi_{k1} = \tilde \delta_k \tilde \pi_{k1}$. Using the previous approach, with $\tau_1=\tau_2=1$ and $\tau_3<1$, with $h=2,\ldots,M$, we have for $t=1,\ldots,K$
\begin{align*}
(1-\varepsilon_h) (1-\varepsilon_1) \sum_{k=1}^K\delta_k \pi_{k1} & \left(\bA_k[1,1](1-\varepsilon_1)\right)^{t-2} \bA_k[1,h] =\\
&(1-\tilde \varepsilon_h) (1-\varepsilon_1) \sum_{k=1}^K\delta_k \pi_{k1}  \left(\bA_k[1,1](1-\varepsilon_1)\right)^{t-2} \tilde \bA_k[1,h],
\end{align*}
and thus $\bA_k[1,h]=\tilde \bA_k[1,h]$ and $\varepsilon_h = \tilde \varepsilon_h$. Similarly, taking $\tau_2<1$ and $\tau_1=\tau_3=1$, we have  $\bA_k[h,1]=\tilde \bA_k[h,1]$. Finally, we have $\bA_k[h,h']=\tilde \bA_k[h,h']$ by increasing $h$ and $h'$, by noting that with suitable choices of $\tau_1$, $\tau_2$ and $\tau_3$, we have for $t=1,\ldots,K$
\begin{align*}
  \sum_{k=1}^K\delta_k \pi_{k1}  \left(\bA_k[1,1](1-\varepsilon_1)\right)^{t-2}& \bA_k[1,h]\bA_k[h,h']\bA_k[h',1] = \\ &\sum_{k=1}^K\tilde \delta_k \tilde \pi_{k1}  \left(\bA_k[1,1](1-\varepsilon_1)\right)^{t-2} \bA_k[1,h]\tilde \bA_k[h,h']\bA_k[h',1].
\end{align*}

\subsection{Proofs of the two technical lemmas}
\begin{proof}[Proof of Lemma~\ref{lem:tech1}]
Since $a_k \neq a_{k'}$ and $\tilde a_k \neq \tilde a_{k'}$, then $N_0$, $N_1$, $\tilde N_0$ and $\tilde N_1$ are Vandermonde matrices and thus are invertible. Therefore, we have $w = N_0^{-1}   \tilde N_0 \tilde w =  N_1^{-1}  \tilde N_1 \tilde w,$ and thus 
\begin{equation*}
( N_0^{-1}   \tilde N_0 -  N_1^{-1}   \tilde N_1) \tilde w=0,
\end{equation*}
or similarly for  $u\in\{1,\ldots,K\}$,
\begin{equation*}
\sum_{k=1}^K a_k^u w_k = \sum_{k=1}^K \tilde a_k^u \tilde w_k.
\end{equation*}
Since the previous equation holds for any $\tilde w$, we have $N_0^{-1}  \tilde N_0=N_1^{-1}  \tilde N_1$. Moreover, we have $N_1=N_0D$ and $\tilde N_1 = \tilde N_0 \tilde D$ where $D=\text{diag}(a_1^K,\ldots,a_K^K)$ and $\tilde D =\mathrm{diag} (\tilde a_1^K,\ldots,\tilde a_K^K)$. Denoting $R=N_0^{-1} \tilde N_0$, $DR=R\tilde D$ and then for $u\in\{1,\ldots,K\}$ and $k\in\{1,\ldots,K\}$
\begin{equation} \label{eq:Dmat}
 a_u^K R[u,k] = \tilde a_k^K R[u,k].
\end{equation}
We now show that $D=\tilde D$ and, $w=\tilde w$, and hence $R=I_K$ and $\tilde N_0=N_0$, where $I_K$ is the identity matrix of size $K$. First we show that $a_1=\tilde a_1$ and $w_1=\tilde w_1$.
\begin{itemize}
\item If $R[1,j]\neq 0$, \eqref{eq:Dmat} implies that $a_1^K R[1,j]=\tilde a_j^K R[1,j]$ and thus $a_1=\tilde a_j$. However, this is impossible because $a_1\geq \tilde a_1 > a_j$ for $j\in\{2,\ldots,K\}$. Hence, we have $R[1,j]=0$ for $j=2,\ldots K$.
\item Noting that $R$ is a product of two invertible matrices, $R$ is invertible. Therefore, $R[1,1]\neq 0$ because $R[1,j]=0$ for $j=2,\ldots K$. Hence, we have $a_1=\tilde a_1$.
\item Note that $R[1,1] = \sum_{k=1}^K (N_0^{-1})[1,k] \tilde N_0[k,1]$ and that $\tilde N_0[k,1] = \tilde a_1^k = a_1^k = N_0[k,1]$. Therefore, we have $R[1,1] = \sum_{k=1}^K (N_0^{-1})[1,k] N_0[k,1] = (N_0^{-1}N_0)[1,1]=1$.
\item For $j=2,\ldots K$,  $a_1 > a_j$ so we have $R[j,1]=0$, because $a_1=\tilde a_1$.
\item Because $w=R\tilde w$, we have $w_1=\tilde w_1$. 
\end{itemize}
Equality of $a_k=\tilde a_k$ and $w_k=\tilde w_k$ can be shown recursively for $k=2,\ldots,K$ using the same reasoning.
\end{proof}

\section{Probabilities of misclassification} \label{app:errorclassif}
\subsection{Technical lemmas}
This section presents some notations and three lemmas which are used for the proof of Theorem~\ref{thm:conv}. 
The technical lemmas discuss the concentration of the frequency of the observation $y_{i(t)}$ in a region of interest $W$, give an upper bound of $p(\ty_i \mid Z_{ik}=1)$ and a concentration result of the ratio of $\frac{p(\tilde\tx_{ik},\ty_i \mid Z_{ik}=1)}{p(\tilde\tx_{ik_0},\ty_i \mid Z_{ik_0}=1)}$, where $\tilde x_{ik} = \argmax_{\tx_i\in\mathcal{X}} p(\tx_i,\ty_i\mid Z_{ik}=1)$ is the estimator of the latent states conditionally on the observation $\ty_i$ and on component $k$ obtained by applying the \emph{maximum a posteriori} rule with the true parameter $\btheta$. The proof of the lemmas uses two concentration results for hidden Markov chains given by \citet{Kontorovich2014} and by \citet{LeonAAP2004}.

\textbf{Preliminaries} Let $\tv_{ik(t)}=(v_{ik(t)h\ell};h=1,\ldots,M;\ell=1,\ldots,M)$ with $\sloppy v_{ik(t)h\ell}=x_{ik(t-1)h} x_{ik(t)\ell}$ and $\tilde{\tv}_{ik(t)}=(\tilde v_{ik(t)h\ell};h=1,\ldots,M;\ell=1,\ldots,M)$ with $\tilde v_{ik(t)h\ell}=\tilde x_{ik(t-1)h} \tilde x_{ik(t)\ell}$. In the following, $\mathbb{P}_0 \left(\cdot\right) = \mathbb{P} \left( \cdot \mid Z_{ik_0}=1 \right)$ by considering the true parameters.

\begin{remark}\label{rem:Markov}
For any $k=1,\ldots,g$, $\tV_{ik(t)}$ is a finite, ergodic and reversible Markov chain with $M^2$ states and transition matrix $\bP_k$ with general term defined for any $(h_1,h_2,h_3,h_4) \in M^4$ by 
$$
\bP_k[(h_1-1)M + h_2, (h_3-1)M + h_4] = \mathbb{P}(V_{ik(t)h\ell} = 1 ) = \left\{ \begin{array}{rl}
0 & \text{if } h_2 \neq h_3 \\
\bA_k[h_2,h_4] & \text{ otherwise } 
\end{array} \right..
$$
Moreover, the non-zero eigenvalues of $\bP_k$ are the non-zero eigenvalues of $\bA_{k}$ and the eigenvectors of $\bP_k$ are obtained from the eigenvectors of $\bA_{k}$.
\end{remark}

\begin{theorem}[\citet{Kontorovich2014}] \label{thm:Kontorovich2014}
Let $U_{(1)},U_{(2)},\ldots$ be a stationary $\mathbb{N}$-valued $(G,\eta)$-geometrically ergodic Markov or hidden Markov chain, and consider the occupation frequency
$$
\hat\rho(E) = \frac{1}{T} \sum_{t=1}^T \mathbf{1}_{\{U_{(t)}\in E\}}, \quad E \subset \mathbb{N}.	 
$$
When $\sum_{u\in\mathbb{N}} \sqrt{\rho_u} <\infty$ with $\rho_u=\mathbb{P}(U_{(1)}=u)$, then for any $\varepsilon>0$
$$
\mathbb{P}\left( \text{sup}_{E \subset \mathbb{N}} \mid \rho(E) - \hat{\rho}(E)\mid > \varepsilon + \gamma_T(G,\eta)\sum_{y\in\mathbb{N}} \sqrt{\rho_y} \right)
\leq
e^{-\frac{T}{2G^2}(1-\eta)^2\varepsilon^2},
$$
where
$$
\gamma_T(G,\eta) = \frac{1}{2} \sqrt{\frac{1 + 2G\eta}{T(1-\eta)}}.
$$
\end{theorem}

\begin{theorem}[\citet{LeonAAP2004}] \label{thm:LeonAAP2004}
For all pairs $(V,f)$, such that $V=(V_{(1)},\ldots,V_{(T)})$ is a finite, ergodic and reversible Markov chain in stationary state with the second-largest eigenvalue $\lambda$ and $f$ is a function taking values in $[0,1]$ such that $\mathbb{E}[f(V_{(t)})] < \infty$, the following bounds, with $\lambda_0=\max (0,\lambda)$, hold for all $s>0$ such that $\mathbb{E}[f(V_{(t)})] +s<1$ and all time $T$
\begin{equation*}
\mathbb{P}\left(\sum_{t=1}^T f(V_{(t)}) \geq (\mathbb{E}[f(V_{(1)})]+ s)T \mid Z_{ik_0}=1 \right) \leq 
\exp \left(-2 \frac{1-\lambda_0}{1 + \lambda_0} T s^2  \right).
\end{equation*}
\end{theorem}

\textbf{Concentration of the frequency of the observations in $W$}
Let $W \subset \mathbb{R}^+$ be the subset of $\mathbb{R}^+$ where the estimator of $x_{i(t)}$ obtained by the \emph{maximum a posteriori} rule is sensitive to $x_{i(t-1)}$ and  $x_{i(t)}$ conditionally on $y_{i(t)}$ and component $k$. We define
$$
W=\{u \in \mathbb{R}^+:\; \text{card}(\cup_{k=1}^g E_k(u))\geq 2 \},
$$
where
$$
E_k(u) = \{h_2:\; \exists(h_1,h_3),\; h_2 = \argmax e_k(u;h_1,h_2,h_3)\},
$$
and
$$
e_k(u;h_1,h_2,h_3) = \bA_k[h_1,h_2] \bA_k[h_2,h_3] g(u;\blambda_{h_2}).
$$

\begin{lemma} \label{lem:techn1}
Let $\rho_{k_0}=\mathbb{P}_0(Y_{i(2)} \in W)$ and $\hat\rho_{k_0} = \sum_{t=1}^T \mathbf{1}_{\{y_{i(t)}\in W\}}$. For any $\delta_1>\frac{1}{\sqrt{2T}}$,
$$
\mathbb{P}_0(\hat\rho_{k_0} < \rho_{k_0} - \delta_1) \leq e^{-T c_1 },
$$
$\delta_1 =  \varepsilon + \frac{1}{\sqrt{2T}} $ and $c_1=\frac{1}{2}(\delta_1 - \frac{1}{\sqrt{2T}})^2>0$. Moreover, $\hat\rho_{k_0}$ is a consistent estimate of $\rho_{k_0}$ because the marginal distribution of $Y_{i(t)}$ is the same for any $t$, and thus $\rho_{k_0}=\mathbb{P}_0(Y_{i(t)} \in W)$ for any $t$.
\end{lemma}

\begin{proof}[Proof of Lemma~\ref{lem:techn1}]
We have,
\[
\mathbb{P}\left(  \mid \rho_{k_0} - \hat{\rho}_{k_0}\mid > \varepsilon + \frac{1}{\sqrt{2T}}  \right)
\leq
\mathbb{P}\left(  \mid \rho_{k_0} - \hat{\rho}_{k_0}\mid > \varepsilon + \frac{1}{2\sqrt{T}} (\sqrt{\rho_{k_0}} + \sqrt{1-\rho_{k_0}}) \right).
\]
Let $U_{(t)}=\mathbf{1}_{\{y_{i(t)}\in W\}}$. Then, for any $k=1,\ldots,g$, $U_{(1)},\ldots,U_{(T)}$ is a stationary $\{0,1\}$-valued $(1,0)$-geometrically ergodic hidden Markov chain conditionally on component $k$. Hence, by Theorem~\ref{thm:Kontorovich2014},
\[
\mathbb{P}\left(  \mid \rho_{k_0} - \hat{\rho}_{k_0}\mid > \varepsilon + \frac{1}{2\sqrt{T}} (\sqrt{\rho_{k_0}} + \sqrt{1-\rho_{k_0}}) \right)
\leq
e^{-\frac{T}{2} \varepsilon^2}.
\]
We can conclude that
\[
\mathbb{P}\left( \hat{\rho}_{k_0}  < \rho_{k_0} - \delta_1   \right)\leq
e^{-Tc_1},
\]
$\delta_1 =  \varepsilon + \frac{1}{\sqrt{2T}} $ and $c_1=\frac{1}{2}(\delta_1 - \frac{1}{\sqrt{2T}})^2$.
\end{proof}

\textbf{Upper-bound of the conditional probability of $\ty_i$ given $Z_{ik}=1$}
Let $\gamma$ and $\bar\gamma$ be upper-bounds of the ratio $\frac{p(\tilde x_{i(t-1)},x_{i(t)},\tilde x_{i(t+1)},y_{i(t)}\mid Z_{ik}=1)}{p(\tilde x_{i(t-1)},\tilde x_{i(t)},\tilde x_{i(t+1)},y_{i(t)}\mid Z_{ik}=1)}$ when $y_{i(t)}\in W$ and $y_{i(t)}\notin W$ respectively. Thus, $\gamma=\max_k \max_{h_1,h_2,h_3,h_4} \frac{A_k[h_1,h_2]}{A_k[h_3,h_4]}$ and $\bar\gamma$ permit to upper bound the ratio between the likelihood computed for any $(\tx_i,\ty_i)$ given $Z_{ik}=1$ and the likelihood computed with $(\tilde\tx_{ik},\ty_i)$  given $Z_{ik}=1$.
We have, if $y_{i(t)}\in W$,
\begin{align*}
\frac{p(\tilde x_{i(t-1)},x_{i(t)},\tilde x_{i(t+1)},y_{i(t)}\mid Z_{ik}=1)}{p(\tilde x_{i(t-1)},\tilde x_{i(t)},\tilde x_{i(t+1)},y_{i(t)}\mid Z_{ik}=1)}\leq
 \max_{u\in W} \max_{h_2\in E_k(u), h_{2'}\in E_k(u), h_2\neq h_{2'}} \frac{\displaystyle\max_{ (h_1,h_3)}e_k(u;h_1,h_2,h_3)}{\displaystyle\min_{(h_{1'},h_{3'}) \in \mathbf{e}_{k}(u;h_{2'})}  e_k(u;h_{1'},h_{2'},h_{3'})} \leq \gamma,
\end{align*}
where $\mathbf{e}_{k}(u;h_2) = \{(h_1,h_3): h_2 = \argmax e_k(u;h_1,h_2,h_3)\}$.
Moreover, we have, if $y_{i(t)}\not\in W$,
\[
\frac{p(\tilde x_{i(t-1)},x_{i(t)},\tilde x_{i(t+1)},y_{i(t)}\mid Z_{ik}=1)}{p(\tilde x_{i(t-1)},\tilde x_{i(t)},\tilde x_{i(t+1)},y_{i(t)}\mid Z_{ik}=1)}\leq
 \max_{u\notin W} \frac{\max_{h_2} \max_{(h_1,h_3) \notin \mathbf{e}_{k}(u;h_2)} e_k(u;h_1,h_2,h_3) }{\max_{h_2  \in E_k(u)}\min_{(h_{1'},h_{3'}) \in \mathbf{e}_{k}(u;h_{2'})}  e_k(u;h_{1'},h_{2'},h_{3'}))} = \bar\gamma.
\]
 Note that $\gamma\geq 1$ and $\bar\gamma < 1$.

\begin{lemma} \label{lem:techn1bis}
We have, for any $k=1,\ldots,g$,
$$
\log p(\ty_i\mid Z_{ik}=1)  \leq 
\log p(\tilde \tx_{ik}, \ty_i \mid Z_{ik}=1 ) + 
T \log (\tilde\gamma + \bar\gamma) + T \hat\rho_{k_0} c_{2} + \log \tilde\gamma + \log \left( 2M\gamma \max_{h,\ell} \frac{\pi_{kh}}{\pi_{k\ell}} \right),
$$
where $c_{2} = 1+\frac{\gamma}{1+\bar\gamma}$ and $\tilde\gamma=\max(2,\gamma)$.
\end{lemma}

\begin{proof}[Proof of Lemma~\ref{lem:techn1bis}]
By definition, we have
$$
p(\ty_i\mid Z_{ik}=1) = p(\tilde \tx_{ik}, \ty_i \mid Z_{ik}=1 ) \sum_{\tx \in \mathcal{X}} \frac{p(\tx, \ty_i \mid Z_{ik}=1 )}{p(\tilde \tx_{ik}, \ty_i \mid Z_{ik}=1 )}.
$$
Let $B_p(\tilde \tx_{ik}) =\{ \tx :\; \mid\mid \tx - \tx_{ik} \mid\mid_0 = p \}$, then
$$
\sum_{\tx \in \mathcal{X}} \frac{p(\tx, \ty_i \mid Z_{ik}=1 )}{p(\tilde \tx_{ik}, \ty_i \mid Z_{ik}=1 )}
= \sum_{p=0}^{T+1} \sum_{\tx \in B_p(\tilde \tx_{ik}) } \frac{p(\tx, \ty_i \mid Z_{ik}=1 )}{p(\tilde \tx_{ik}, \ty_i \mid Z_{ik}=1 )}.
$$
Remark that $$\frac{p(\tx_{i(0)}, \ty_{i(0)}\mid Z_{ik}=1,\tx_{i(1)},\ty_{i(1)})}{p(\tilde\tx_{i(0)}, \ty_{i(0)}\mid Z_{ik}=1,\tilde\tx_{i(1)},\ty_{i(1)})} <\gamma \max_{h,\ell} \frac{\pi_{kh}}{\pi_{k\ell}}.$$
Moreover, we observe $T_W=T\hat\rho_{k_0}$ elements of the sequence $y_{i(1)},\ldots,y_{i(T)}$ which belongs to $W$. We have
\begin{align*}
\sum_{\tx \in \mathcal{X}} \frac{p(\tx, \ty_i \mid Z_{ik}=1 )}{p(\tilde \tx_{ik}, \ty_i \mid Z_{ik}=1 )}
 \leq &
\left(M\gamma \max_{h,\ell} \frac{\pi_{kh}}{\pi_{k\ell}}\right)\sum_{p=0}^{T} \sum_{r=0}^p  \binom{T_W}{\min (r, T_W)} \binom{T-T_W}{\min (u, T-T_W)} \gamma^r \bar\gamma^u \\
= & \left(M\gamma \max_{h,\ell} \frac{\pi_{kh}}{\pi_{k\ell}}\right) \left( \sum_{r=0}^{T_W} \binom{T_W}{r} \gamma^r  \sum_{u=0}^{T-r}  \binom{T-T_W}{u}  \bar\gamma^u \right.\\
& \; \; + \left.\sum_{r=1+T_W}^{T}  \sum_{u=0}^{T-r} \binom{T-T_W}{\min(u,T-T_W)} \gamma^r \bar\gamma^u\right).
\end{align*}
We have 
$$
\sum_{r=0}^{T_W} \binom{T_W}{r} \gamma^r  \sum_{u=0}^{T-r}  \binom{T-T_W}{u}  \bar\gamma^u=(1+\bar\gamma)^{T}\left(1+\frac{\gamma}{1+\bar\gamma}\right)^{T_W},
$$
and
$$
\sum_{r=1+T_W}^{T}  \sum_{u=0}^{T-r} \binom{T-T_W}{\min(u,T-T_W)} \gamma^r \bar\gamma^u 
\leq  (\tilde\gamma + \bar\gamma )^{T} \tilde\gamma \left(\frac{\tilde\gamma}{\bar\gamma + \tilde\gamma}\right)^{T_W},
$$
where $\tilde\gamma=\max(2,\gamma)$. Noting that $1+\bar\gamma<\tilde\gamma + \bar\gamma $ and $1+\frac{\gamma}{1+\bar\gamma}>\frac{\tilde\gamma}{\bar\gamma + \tilde \gamma}$, we have
$$
\log p(\ty_i\mid Z_{ik}=1)  \leq 
\log p(\tilde \tx_{ik}, \ty_i \mid Z_{ik}=1 ) + 
T \log (\tilde\gamma + \bar\gamma) + T \hat\rho_{k_0} c_{2} + \log \tilde\gamma + \log \left( 2M\gamma \max_{h,\ell} \frac{\pi_{kh}}{\pi_{k\ell}} \right),
$$
where $c_{2} = 1+\frac{\gamma}{1+\bar\gamma}$. Note that $\gamma+1>c_2>1$.

\end{proof}

\textbf{Concentration of the ratio of complete-data likelihood}

\begin{lemma}\label{lem:techn2}
For any $k\neq k_0$ and for any $\delta_3$ such that $-\zeta<\delta_3< u_{kk_0}$,  we have 
$$
\mathbb{P}_0 \left(\frac{1}{T} \sum_{t=1}^T \sum_{h=1}^M \sum_{\ell=1}^M v_{i(t)h\ell} \log \left(\frac{A_k[h,\ell]}{A_{k_0}[h,\ell]}\right) > \delta_3  \right)
\leq
\exp \left(- T c_3\right),
$$
where $c_3=2 \frac{1-\bar\nu_2(\bA_{k_0})}{1 + \bar\nu_2(\bA_{k_0})}   s^2 >0$ and $s=\frac{\delta_3}{\omega_{kk_0}} + \frac{1}{\omega_{kk_0}} \sum_{h=1}^M \sum_{\ell=1}^M \pi_{k_0h}A_{k_0}[h,\ell] \log \left(\frac{A_{k_0}[h,\ell]}{A_{k}[h,\ell]}\right)$.
\end{lemma}

\begin{proof}[Proof of Lemma~\ref{lem:techn2}]
 Let $f(\cdot)\in[0,1]$ defined by $$f(\tv_{i(t)})= \frac{1}{\omega_{kk_0}}\left( \sum_{h=1}^M \sum_{\ell=1}^M v_{i(t)h\ell} \log \left(\frac{A_k[h,\ell]}{A_{k_0}[h,\ell]}\right) + u_{k_0 k}\right),$$ where $\omega_{kk_0}=u_{kk_0} + u_{k_0k}$, $u_{kk_0} = \max_{(h,\ell)} \log \left(\frac{A_k[h,\ell]}{A_{k_0}[h,\ell]}\right)$. Denoting $\mathbb{E}_0[\cdot] = \mathbb{E}[\cdot \mid Z_{ik_0}=1]$ the conditional expectation computed with the true parameters, we have, for $t=1,\ldots,T$,
$$
\mathbb{E}_0\left[ f(\tV_{i(t)}) \right] = \frac{1}{\omega_{kk_0}} \sum_{h=1}^M \sum_{\ell=1}^M \pi_{k_0h}A_{k_0}[h,\ell] \left(\log \left(\frac{A_k[h,\ell]}{A_{k_0}[h,\ell]}\right) +  u_{k_0k}\right).
$$
Therefore, we have
\begin{align*}
\mathbb{P}_0 \left( \sum_{t=1}^T \sum_{h=1}^M \sum_{\ell=1}^M v_{i(t)h\ell} \log \left(\frac{A_k[h,\ell]}{A_{k_0}[h,\ell]}\right) > \delta_2  \right)
& = \mathbb{P}_0 \left( \sum_{t=1}^{T} f(\tv_{i(t)}) > \frac{\delta_2 + T u_{k_0k}}{\omega_{kk_0}}  \right) \\
& = \mathbb{P}_0 \left( \sum_{t=1}^{T} f(\tv_{i(t)}) > T (\mathbb{E}[f(\tV_{i(1)})]+ s) \right),
\end{align*}
where $s=\frac{\delta_2}{T\omega_{kk_0}} + \frac{1}{\omega_{kk_0}} \sum_{h=1}^M \sum_{\ell=1}^M \pi_{k_0h}A_{k_0}[h,\ell] \log \left(\frac{A_{k_0}[h,\ell]}{A_{k}[h,\ell]}\right)$. 
 
Note that $\omega_{kk_0}>0$ and that, by Assumption~\ref{ass:kullback}, $\sum_{h=1}^M \sum_{\ell=1}^M \pi_{k_0h}A_{k_0}[h,\ell] \log \left(\frac{A_{k_0}[h,\ell]}{A_{k}[h,\ell]}\right)>\zeta>0$ because it is a weighted sum of $M$ Kullback-Leibler divergences. Thus, if $-T\zeta < \delta_2$ then  $s>0$.
Moreover, if $\delta_2<Tu_{kk_0}$,  then $\mathbb{E}[f(\tV_{i(1)})]+ s <1$. Assumption~\ref{ass:chaine} and Remark~\ref{rem:Markov} imply that $\bar\nu_2(\bA_{k_0})$ is the maximum between zero and the second-largest eigenvalue of reversible Markov chain of $\tV_{i(t)}$. Therefore, using Theorem~\ref{thm:LeonAAP2004}, we have for any $\delta_3$ such that $-\zeta<\delta_3<u_{kk_0}$,
\[
\mathbb{P}_0 \left(\frac{1}{T} \sum_{t=1}^T \sum_{h=1}^M \sum_{\ell=1}^M v_{i(t)h\ell} \log \left(\frac{A_k[h,\ell]}{A_{k_0}[h,\ell]}\right) > \delta_3  \right)
\leq
\exp \left(-T c_3 \right),
\]
where $c_3=2 \frac{1-\bar\nu_2(\bA_{k_0})}{1 + \bar\nu_2(\bA_{k_0})} s^2 $ and $s=\frac{\delta_3}{\omega_{kk_0}} + \frac{1}{\omega_{kk_0}} \sum_{h=1}^M \sum_{\ell=1}^M \pi_{k_0h}A_{k_0}[h,\ell] \log \left(\frac{A_{k_0}[h,\ell]}{A_{k}[h,\ell]}\right)$.
\end{proof}

\subsection{Proof of Theorem~\ref{thm:conv}}
Noting that $\mathbb{P}(Z_{ik}=1 \mid \ty_i)\propto \delta_k p(\ty_i \mid Z_{ik}=1)$ and using Lemma~\ref{lem:techn1bis}, we have
\begin{equation*}
\begin{aligned}
\mathbb{P}_0\left(  \frac{\mathbb{P}(Z_{ik}=1 \mid \ty_i)}{\mathbb{P}(Z_{ik_0}=1 \mid \ty_i)} > a \right) \leq \mathbb{P}_0\left( \log \frac{p(\tilde\tx_{ik},\ty_i\mid Z_{ik}=1)}{p(\tilde\tx_{ik_0},\ty_i\mid Z_{ik_0}=1)} > - \log \frac{\delta_k}{a\delta_{k_0}} - \log  \left(2M\tilde\gamma \gamma \max_{h,\ell} \frac{\pi_{kh}}{\pi_{k\ell}}\right)     \right. \\
\left.- T \log (\tilde\gamma + \bar\gamma) - T \hat\rho_{k_0} c_{2}  \right).
\end{aligned}
\end{equation*}
Moreover, 
$$
\log \frac{p(\tilde\tx_{ik},\ty_i\mid Z_{ik}=1)}{p(\tilde\tx_{ik_0},\ty_i\mid Z_{ik_0}=1)}  = 
\sum_{t=1}^T \left( d_{k1(t)}+d_{k2(t)} \right) + \sum_{h=1}^M \tilde x_{ik(1)h}\log \pi_{kh} - \tilde x_{ik_0(1)h} \log \pi_{k_0h} ,
$$
where 
\[
d_{k1(t)}= \sum_{h=1}^M \sum_{\ell=1}^M(\tilde v_{ik(t)h\ell} - \tilde v_{ik_0(t)h\ell}) \log \left( \bA_{k_0}[h,\ell]g_\ell(y_{i(t)};\blambda_\ell,\varepsilon_\ell) \right),
\]
 and 
 \[
 d_{k2(t)}= \sum_{h=1}^M \sum_{\ell=1}^M \tilde v_{ik(t)h\ell} \log \frac{\bA_k[h,\ell]}{\bA_{k_0}[h,\ell]}.
 \]
Therefore, we have
$$
\mathbb{P}_0\left(  \frac{\mathbb{P}(Z_{ik}=1 \mid \ty_i)}{\mathbb{P}(Z_{ik_0}=1 \mid \ty_i)} > a \right) \leq \mathbb{P}_0\left( \frac{1}{T}\sum_{t=1}^T \left( d_{k1(t)}+d_{k2(t)} \right) > - \frac{c_4}{T}  - \log (\tilde\gamma + \bar\gamma) - \hat\rho_{k_0} c_{2}  \right),
$$
with $c_4= \log \frac{\delta_k}{a\delta_{k_0}} + \log  \left(2M\tilde\gamma \gamma \max_{h,\ell} \frac{\pi_{kh}}{\pi_{k\ell}}\right) + \max_{k,k_0,h,\ell} \log \frac{\pi_{kh}}{\pi_{k_0\ell}}$.
By definition of $W$, we have $\tilde v_{ik(t)h\ell}=\tilde v_{ik_0(t)h\ell}$ if $y_{i(t)}\notin W$. Moreover,  because $\tilde \tv_{ik_0}$ is the maximum \emph{a posteriori} rule, if $y_{i(t)}\in W$, then $d_{k1(t)}< \gamma$. Therefore, we have 
\begin{align*}
\mathbb{P}_0\left(  \frac{\mathbb{P}(Z_{ik}=1 \mid \ty_i)}{\mathbb{P}(Z_{ik_0}=1 \mid \ty_i)} > a \right) &\leq \mathbb{P}_0\left( \frac{1}{T}\sum_{t=1}^T d_{k2(t)} > -\frac{c_4}{T} -  (\gamma + c_2) \hat\rho_{k_0} -  \log(\tilde\gamma + \bar\gamma)  \right). 
\end{align*}
Hence, we have,  
\begin{equation*}
\begin{aligned}
\mathbb{P}_0\left(  \frac{\mathbb{P}(Z_{ik}=1 \mid \ty_i)}{\mathbb{P}(Z_{ik_0}=1 \mid \ty_i)} > a \right)  
\leq \mathbb{P}_0\left(   \hat \rho_{k_0} > \rho_{k_0} + \delta_1   \right) +  \mathbb{P}_0\left( \frac{1}{T}\sum_{t=1}^Td_{k2}(t) > -\frac{c_4}{T}  -  \log(\tilde\gamma + \bar\gamma)\right.\\
\left. -  (\gamma + c_2) (\rho_{k_0} + \delta_1) \right).
\end{aligned}
\end{equation*}
Using Lemma~\ref{lem:techn1}, if $\delta_1>\frac{1}{\sqrt{2T}}$, the first term of the right side of the previous equation can be upper bounded by $e^{-Tc_1}$  with  $c_1=\frac{1}{2}(\delta_1 - \frac{1}{\sqrt{2T}})^2$.

Using Lemma~\ref{lem:techn2}, the second term of the right-hand side of the previous equation can be upper bounded by $e^{-T c_3}$ with  $c_3 =2 \frac{1-\bar\nu_2(\bA_{k_0})}{1 + \bar\nu_2(\bA_{k_0})} s^2$,  where $s=\frac{\delta_3}{\omega_{kk_0}} + \frac{1}{\omega_{kk_0}} \sum_{h=1}^M \sum_{\ell=1}^M \pi_{k_0h}A_{k_0}[h,\ell] \log \left(\frac{A_{k_0}[h,\ell]}{A_{k}[h,\ell]}\right)$ and $\delta_3 = -\frac{c_4}{T}  -  \log(\tilde\gamma + \bar\gamma) -  (\gamma + c_2) (\rho_{k_0} + \delta_1) $,  if $\delta_3$ is such that $-\zeta<\delta_3<u_{kk_0}$. Thus, we have the following condition on $\delta_1$
 $$\frac{\zeta - \frac{c_4}{T} - \log(\tilde\gamma + \bar\gamma) }{\gamma + c_2} - \rho_{k_0} > \delta_1 >- \frac{u_{kk_0} + \frac{c_4}{T}+ \log(\tilde\gamma + \bar\gamma) }{\gamma + c_2} - \rho_{k_0}.$$
Noting that $\gamma + c_2<1 + 2\gamma$, the previous upper bound can be satisfied under the following assumption

\begin{assumption}\label{ass:expo}
It holds that
\begin{equation*}
\frac{\zeta- c_4 - \log(\tilde\gamma + \bar\gamma)}{1 + 2 \gamma} - \rho_{k_0} - \frac{1}{\sqrt{2}}> 0,
\end{equation*}
with $c_4= \log \frac{\delta_k}{a\delta_{k_0}} + \log  \left(2M\tilde\gamma \gamma \max_{h,\ell} \frac{\pi_{kh}}{\pi_{k\ell}}\right) + \max_{k,k_0,h,\ell} \log \frac{\pi_{kh}}{\pi_{k_0\ell}}$.
\end{assumption}
For any $a$ such that Assumption~\ref{ass:expo} holds and for any $\delta_1$ with $\frac{1}{\sqrt{2T}}<\delta_1< \frac{\zeta - c_4 - \log(\tilde\gamma + \bar\gamma) }{\gamma + c_2} - \rho_{k_0}$, we have
$$
\mathbb{P}_0\left(   \hat \rho_{k_0} > \rho_{k_0} + \delta_1   \right)\leq \mathcal{O}(e^{-Tc_1}),
$$
and 
$$
\mathbb{P}_0\left( \frac{1}{T}\sum_{t=1}^Td_{k2}(t) > \delta_3 ) \right)\leq \mathcal{O}(e^{-Tc_3}),
$$
with $\delta_3=-\frac{c_4}{T}  -  \log(\tilde\gamma + \bar\gamma) -  (\gamma + c_2) (\rho_{k_0} + \delta_1)$. 
Therefore, there exists $c>0$ such that
$$
\mathbb{P}_0\left(  \frac{\mathbb{P}(Z_{ik}=1 \mid \ty_i)}{\mathbb{P}(Z_{ik_0}=1 \mid \ty_i)} > a \right)  
\leq \mathcal{O}(e^{-Tc}).
$$
If the misclassification error is studied, we should consider $a=1$. Then, a sufficient condition to have the exponential decreasing of the probability of misclassifying an observation is obtained on the basis of Assumption~\ref{ass:expo} with $a=1$.

\section{Details about the conditional distribution} 
\label{sec:conddistribution}
\textbf{Forward formula}
We define
\begin{equation*}
\alpha_{ikhs(t)}(\btheta) = \mathbb{P}(X_{is(t)} = h \mid Z_{ik} =1; \btheta) \, p(y_{is(0)},\ldots,y_{is(t)} \mid X_{is(t)}=h, Z_{ik}=1; \btheta),
\end{equation*}
which measures the probability of the partial sequence $y_{is(0)},\ldots,y_{is(t)}$ and ending up in state $h$ at time $t$ under component $k$. For any $(i,k,h,s)$, we can define $\alpha_{ikhs(t)}$ recursively, as follows,
\begin{align*}
\alpha_{ikhs(à)}(\btheta) & = \pi_{kh} \, p(y_{is(à)};\blambda_h) \\
\forall t \in \{0,\ldots,T_{is}-1\} \quad \alpha_{ik \ell s(t+1)}(\btheta) & = \left( \sum_{h=1}^M A_k[h,\ell] \alpha_{ikhs(t)}(\btheta) \right) p(y_{is(t+1)};\blambda_{h}). 
\end{align*}
Considering independence between the $S_i$ sequences $\ty_{is}$, the pdf of $\ty_i$ under component $k$ is 
\begin{equation*}
p(\ty_i \mid Z_{ik}=1;\btheta) = \prod_{s=1}^{S_i}  \sum_{h=1}^M  \alpha_{ik h s(T_{is})}(\btheta).
\end{equation*}
Therefore,
\begin{equation*}
p(\ty_i;\btheta) = \sum_{k=1}^K \delta_k \left(\prod_{s=1}^{S_i}  \sum_{h=1}^M  \alpha_{ik h s(T_{is})}(\btheta)\right).
\end{equation*}

\textbf{Backward formula}
We define
\begin{equation*}
\beta_{ikhs(t)}(\btheta) = p(y_{is(t+1)},\ldots,y_{is(T_{is})} \mid X_{is(t)}=h, Z_{ik}=1; \btheta),
\end{equation*}
which measures the probability of the ending partial sequence $y_{is(t+1)},\ldots,y_{is(T_{is})}$ given a start in state $h$ at time $t$ under component $k$. We can define $\beta_{ikhs(t)}(\btheta)$ recursively, for any $(i,k,h,s)$, as
\begin{align*}
\beta_{ikhs(T_{is})}(\btheta)&= 1 \\
\forall t\in \{0,\ldots,T_{is}-1\} \quad \beta_{ikhs(t)}(\btheta) &= \sum_{\ell=1}^M A_k[h,\ell] p(y_{i(t+1)};\blambda_{\ell})  \beta_{ik\ell s(t+1)}(\btheta).
\end{align*}
Considering independence between the $S_i$ sequences $\ty_{is}$, the pdf of $\ty_i$ under component $k$ is 
\begin{equation*}
p(\ty_i \mid Z_{ik}=1;\btheta) = \prod_{s=1}^{S_i}  \sum_{h=1}^M  \pi_{kh} \beta_{ikhs(0)}(\btheta) p(y_{i(0)};\blambda_{h}).
\end{equation*}
\begin{equation*}
p(\ty_i;\btheta) = \sum_{k=1}^K \delta_k \left(\prod_{s=1}^{S_i}  \sum_{h=1}^M  \pi_{kh} \beta_{ikhs(0)}(\btheta) p(y_{i(0)};\blambda_{h})\right).
\end{equation*}

\section{Additional simulation results}
\subsection{Additional tables for the analysis of simulated data} 
\label{app:simu}

\begin{table}[H]
\caption{Convergence of the estimators with 1000 replicates: ARI between estimated and true partition, ARI between estimated and true latent states and MSE between the MLE and the true parameters}
\centering
\begin{tabular}{llccccccc}
\hline
 & & \multicolumn{2}{c}{Adjusted Rand index} & \multicolumn{5}{c}{Mean square error} \\
$n$ & $T$ & partition & states & $\bA_k$ & $\varepsilon_h$ & $a_h$ & $b_h$ & $\delta_k$ \\
\hline
  \multicolumn{9}{c}{Case hard ($e=0.75$ and $a_2=3$)}\\
\hline
 10 & 100 & 0.812 & 0.350 & 0.097 & 0.002 & 0.285 & 0.085 & 0.057 \\ 
10 & 500 & 1.000 & 0.385 & 0.031 & 0.000 & 0.077 & 0.024 & 0.052 \\ 
 100 & 100 & 0.870 & 0.377 & 0.032 & 0.000 & 0.087 & 0.028 & 0.007 \\ 
100 & 500 & 1.000 & 0.391 & 0.018 & 0.000 & 0.045 & 0.017 & 0.005 \\ 
\hline
  \multicolumn{9}{c}{Case medium-easy ($e=0.75$ and $a_2=5$)}\\ 
\hline
10 & 100 & 0.996 & 0.661 & 0.018 & 0.001 & 0.253 & 0.031 & 0.052 \\ 
10 & 500 & 0.999 & 0.669 & 0.004 & 0.000 & 0.050 & 0.006 & 0.052 \\ 
100 & 100 & 0.998 & 0.669 & 0.002 & 0.000 & 0.027 & 0.003 & 0.005 \\ 
100 & 500 & 1.000 & 0.671 & 0.000 & 0.000 & 0.005 & 0.001 & 0.005 \\
\hline
  \multicolumn{9}{c}{Case easy ($e=0.90$ and $a_2=5$)}\\ 
\hline
  \hline
10 & 100 & 1.000 & 0.832 & 0.006 & 0.000 & 0.159 & 0.017 & 0.048 \\ 
 10 & 500 & 1.000 & 0.839 & 0.001 & 0.000 & 0.030 & 0.003 & 0.046 \\ 
 100 & 100 & 1.000 & 0.836 & 0.001 & 0.000 & 0.015 & 0.002 & 0.005 \\ 
  100 & 500 & 1.000 & 0.839 & 0.000 & 0.000 & 0.003 & 0.000 & 0.005 \\ 
\hline
\end{tabular}
\end{table}

\begin{table}[H]
 \caption{Convergence of the estimators obtained on 1000 replicates with and without missing data when data are sampled from case hard: ARI between estimated and true partition, ARI between estimated and true latent states and MSE between the MLE and the true parameters}
\centering
\begin{tabular}{lllccccccc}
\hline
 & & &  \multicolumn{2}{c}{Adjusted Rand index} & \multicolumn{5}{c}{Mean square error} \\
$n$ & $T$ & missingness & partition & states & $\bA_k$ & $\varepsilon_h$ & $a_h$ & $b_h$ & $\delta_k$ \\
\hline
10 & 100 & no missingness & 0.812 & 0.350 & 0.097 & 0.002 & 0.285 & 0.085 & 0.057 \\ 
&&  MCAR-1 & 0.753 & 0.342 & 0.103 & 0.002 & 0.360 & 0.092 & 0.061 \\ 
&&  MCAR-2 & 0.690 & 0.333 & 0.120 & 0.003 & 0.320 & 0.091 & 0.068 \\ 
&&  MNAR & 0.542 & 0.303 & 0.152 & 0.004 & 0.456 & 0.109 & 0.080 \\ 
10 & 500 & no missingness & 1.000 & 0.385 & 0.031 & 0.000 & 0.077 & 0.024 & 0.052 \\ 
&&  MCAR-1 & 1.000 & 0.386 & 0.029 & 0.000 & 0.072 & 0.024 & 0.052 \\ 
&&  MCAR-2 & 0.999 & 0.385 & 0.032 & 0.000 & 0.074 & 0.024 & 0.052 \\ 
&&  MNAR & 0.991 & 0.347 & 0.042 & 0.003 & 0.177 & 0.076 & 0.052 \\ 
100 & 100 & no missingness & 0.870 & 0.377 & 0.032 & 0.000 & 0.087 & 0.028 & 0.007 \\ 
&&  MCAR-1 & 0.837 & 0.373 & 0.033 & 0.000 & 0.091 & 0.030 & 0.008 \\ 
&&  MCAR-2 & 0.801 & 0.369 & 0.035 & 0.000 & 0.099 & 0.032 & 0.008 \\ 
&&  MNAR & 0.664 & 0.337 & 0.038 & 0.003 & 0.099 & 0.083 & 0.011 \\ 
100 & 500 & no missingness & 1.000 & 0.391 & 0.018 & 0.000 & 0.045 & 0.017 & 0.005 \\ 
&&  MCAR-1 & 1.000 & 0.391 & 0.018 & 0.000 & 0.050 & 0.017 & 0.005 \\ 
&&  MCAR-2 & 1.000 & 0.391 & 0.018 & 0.000 & 0.044 & 0.017 & 0.005 \\ 
&&  MNAR & 0.996 & 0.355 & 0.023 & 0.003 & 0.109 & 0.073 & 0.005 \\ 
\hline
\end{tabular}
\end{table}

 \begin{table}[H]
 \caption{Convergence of the estimators obtained on 1000 replicates with and without missing data when data are sampled from case medium-easy: ARI between estimated and true partition, ARI between estimated and true latent states and MSE between the MLE and the true parameters}
  \centering
\begin{tabular}{lllccccccc}
\hline
 & & &  \multicolumn{2}{c}{Adjusted Rand index} & \multicolumn{5}{c}{Mean square error} \\
$n$ & $T$ & missingness & partition & states & $\bA_k$ & $\varepsilon_h$ & $a_h$ & $b_h$ & $\delta_k$ \\
\hline
10 & 100 & no missingness  & 0.996 & 0.661 & 0.018 & 0.001 & 0.253 & 0.031 & 0.052 \\ 
&&  MCAR-1 & 0.993 & 0.659 & 0.020 & 0.001 & 0.276 & 0.034 & 0.052 \\ 
&&  MCAR-2 & 0.986 & 0.657 & 0.023 & 0.001 & 0.308 & 0.038 & 0.051 \\ 
&&  MNAR & 0.966 & 0.621 & 0.052 & 0.004 & 0.561 & 0.100 & 0.052 \\ 
10 & 500 & no missingness & 0.999 & 0.669 & 0.004 & 0.000 & 0.050 & 0.006 & 0.052 \\ 
&&  MCAR-1 & 0.999 & 0.669 & 0.004 & 0.000 & 0.051 & 0.006 & 0.052 \\ 
&&  MCAR-2 & 0.999 & 0.668 & 0.004 & 0.000 & 0.052 & 0.007 & 0.052 \\ 
&&  MNAR & 1.000 & 0.636 & 0.021 & 0.003 & 0.217 & 0.085 & 0.052 \\ 
100 & 100 & no missingness & 0.998 & 0.669 & 0.002 & 0.000 & 0.027 & 0.003 & 0.005 \\ 
&&  MCAR-1 & 0.996 & 0.668 & 0.002 & 0.000 & 0.030 & 0.003 & 0.005 \\ 
&&  MCAR-2 & 0.993 & 0.667 & 0.002 & 0.000 & 0.033 & 0.004 & 0.005 \\ 
&&  MNAR & 0.978 & 0.641 & 0.013 & 0.003 & 0.160 & 0.074 & 0.005 \\ 
100 & 500 & no missingness & 1.000 & 0.671 & 0.000 & 0.000 & 0.005 & 0.001 & 0.005 \\ 
&&  MCAR-1 & 1.000 & 0.671 & 0.000 & 0.000 & 0.005 & 0.001 & 0.005 \\ 
&&  MCAR-2 & 1.000 & 0.671 & 0.000 & 0.000 & 0.005 & 0.001 & 0.005 \\ 
&&  MNAR & 1.000 & 0.644 & 0.011 & 0.003 & 0.140 & 0.076 & 0.005 \\ 
\hline
\end{tabular}
\end{table}

 \begin{table}[H]
 \caption{Convergence of the estimators obtained on 1000 replicates with and without missing data when data are sampled from case easy: ARI between estimated and true partition, ARI between estimated and true latent states and MSE between the MLE and the true parameters}
  \centering
\begin{tabular}{lllccccccc}
\hline
 & & &  \multicolumn{2}{c}{Adjusted Rand index} & \multicolumn{5}{c}{Mean square error} \\
$n$ & $T$ & missingness & partition & states & $\bA_k$ & $\varepsilon_h$ & $a_h$ & $b_h$ & $\delta_k$ \\
\hline
10 & 100 & no missingness   & 1.000 & 0.832 & 0.006 & 0.000 & 0.159 & 0.017 & 0.048 \\ 
&&  MCAR-1  & 1.000 & 0.829 & 0.007 & 0.000 & 0.185 & 0.019 & 0.048 \\ 
&&  MCAR-2 & 1.000 & 0.826 & 0.008 & 0.001 & 0.207 & 0.022 & 0.048 \\ 
&&  MNAR & 1.000 & 0.771 & 0.017 & 0.003 & 0.315 & 0.042 & 0.048 \\ 
10 & 500 & no missingness   & 1.000 & 0.839 & 0.001 & 0.000 & 0.030 & 0.003 & 0.046 \\ 
&&  MCAR-1  & 1.000 & 0.838 & 0.001 & 0.000 & 0.031 & 0.003 & 0.046 \\ 
&&  MCAR-2 & 1.000 & 0.837 & 0.001 & 0.000 & 0.031 & 0.003 & 0.046 \\ 
&&  MNAR & 1.000 & 0.781 & 0.007 & 0.003 & 0.119 & 0.030 & 0.046 \\ 
100 & 100 & no missingness   & 1.000 & 0.836 & 0.001 & 0.000 & 0.015 & 0.002 & 0.005 \\ 
&&  MCAR-1  & 1.000 & 0.833 & 0.001 & 0.000 & 0.018 & 0.002 & 0.005 \\ 
&&  MCAR-2 & 1.000 & 0.831 & 0.001 & 0.000 & 0.019 & 0.002 & 0.005 \\ 
&&  MNAR & 1.000 & 0.779 & 0.006 & 0.003 & 0.084 & 0.026 & 0.005 \\ 
100 & 500 & no missingness   & 1.000 & 0.839 & 0.000 & 0.000 & 0.003 & 0.000 & 0.005 \\ 
&&  MCAR-1 & 1.000 & 0.838 & 0.000 & 0.000 & 0.003 & 0.000 & 0.005 \\ 
&&  MCAR-2 & 1.000 & 0.838 & 0.000 & 0.000 & 0.003 & 0.000 & 0.005 \\ 
&&  MNAR & 1.000 & 0.783 & 0.005 & 0.003 & 0.072 & 0.025 & 0.005 \\ 
\hline
\end{tabular}
\end{table}

\subsection{Additional figures for the analysis of the data from \texorpdfstring{\cite{huang2018hidden}}{Huang et al. (2018b)}} 
\label{app:huang}

\begin{figure}[H]
\centering 
\includegraphics[width=0.95\textwidth]{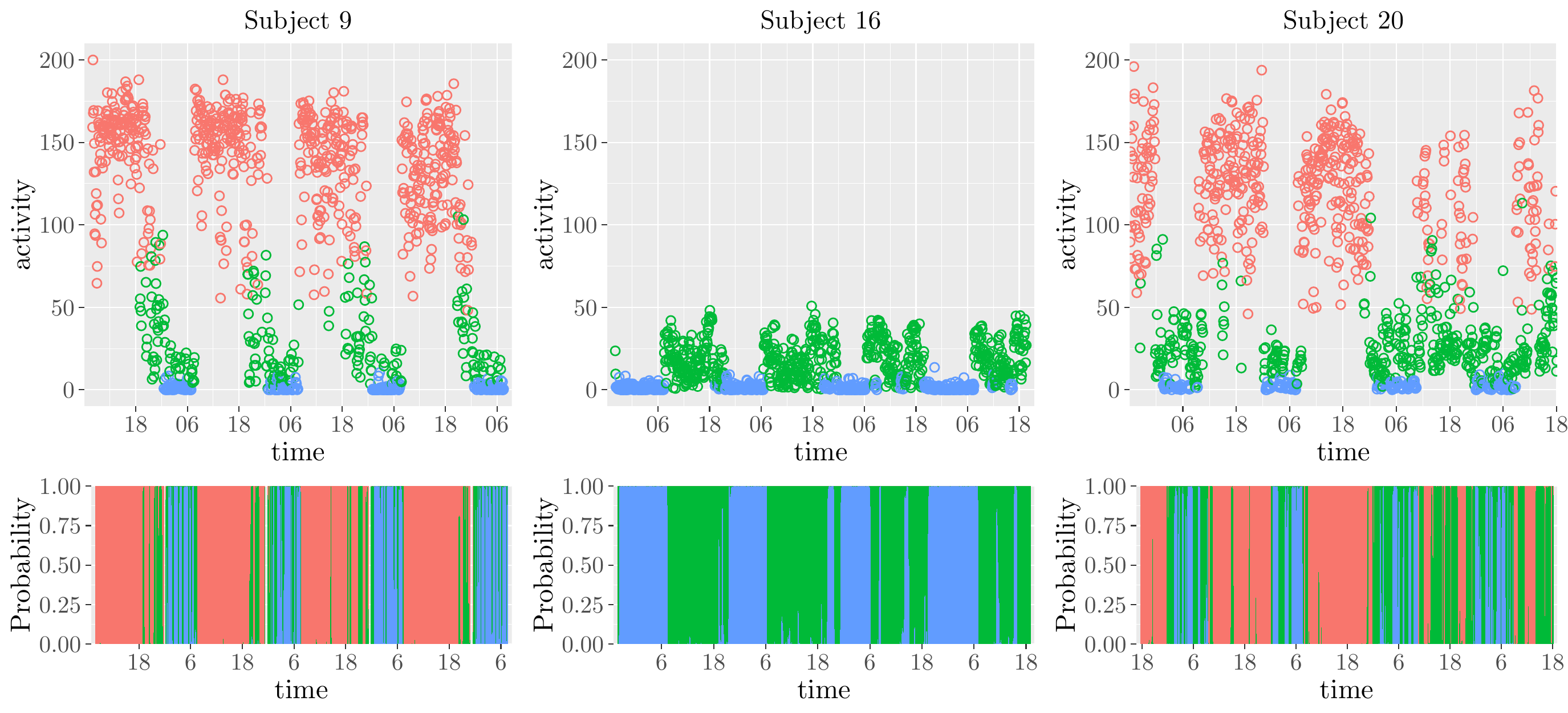}
\includegraphics[width=0.3\textwidth]{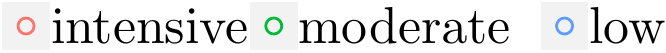}
\caption{State estimation for the three subjects: (top) accelerometer data where color indicates the expected value of $Y_{i(t)}$ conditionally the most likely state and on the most likely component; (bottom) probability of each state at each time.}
\label{fig::huangdata2}
\end{figure}

\begin{figure}[H]
\centering 
\includegraphics[width=0.95\textwidth]{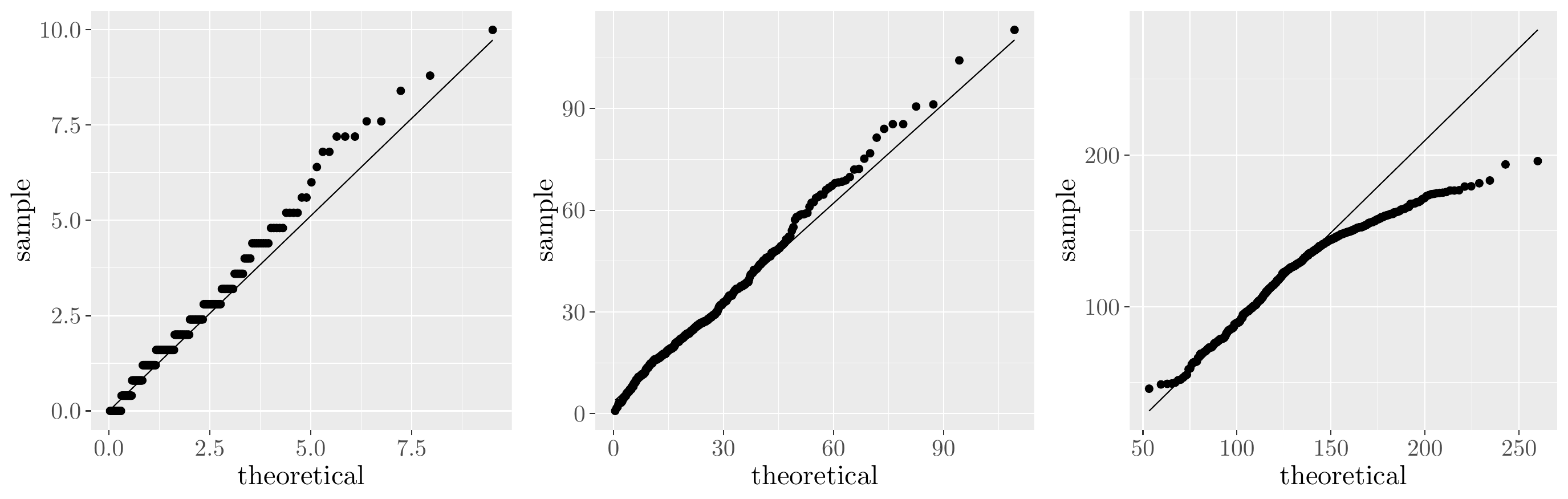}
\includegraphics[width=0.95\textwidth]{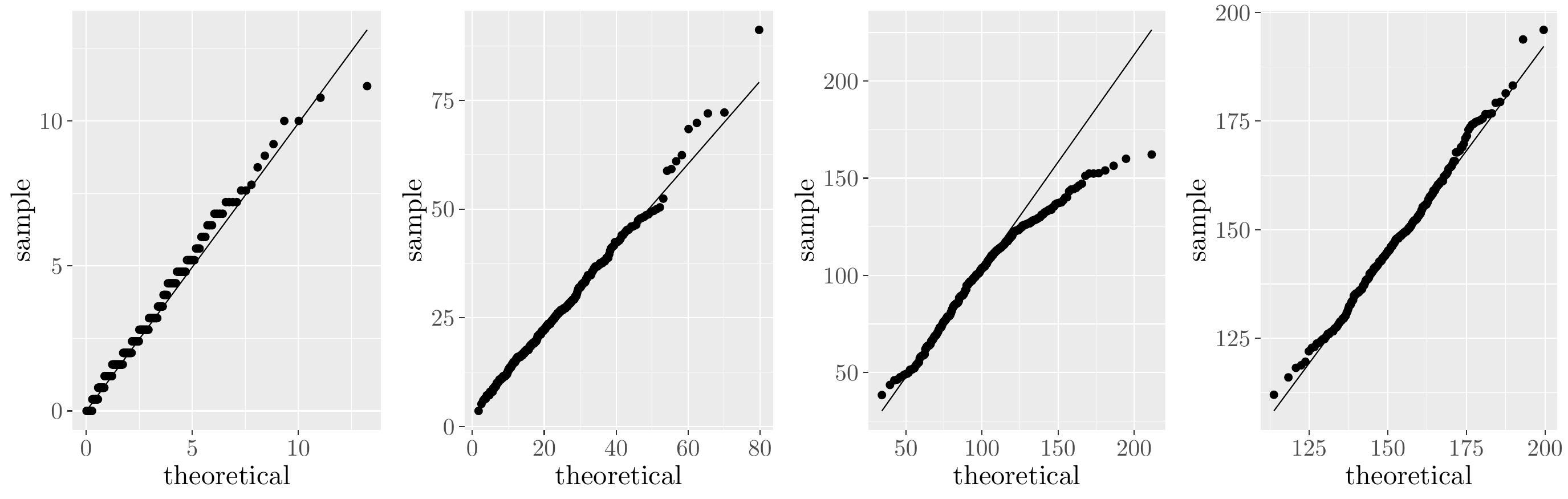}
\caption{QQ-plots of subject 20. Top: from left to right, state 1 to 3. Bottom: Top: from left to right, state 1 to 4.}
\label{fig::huang_qq3}
\end{figure}

\subsection{Additional figures for the analysis of the PAT data}\label{supp:pat}
\begin{figure}
\begin{center}
\subfigure[\emph{Active} class]{\includegraphics[width=0.47\textwidth]{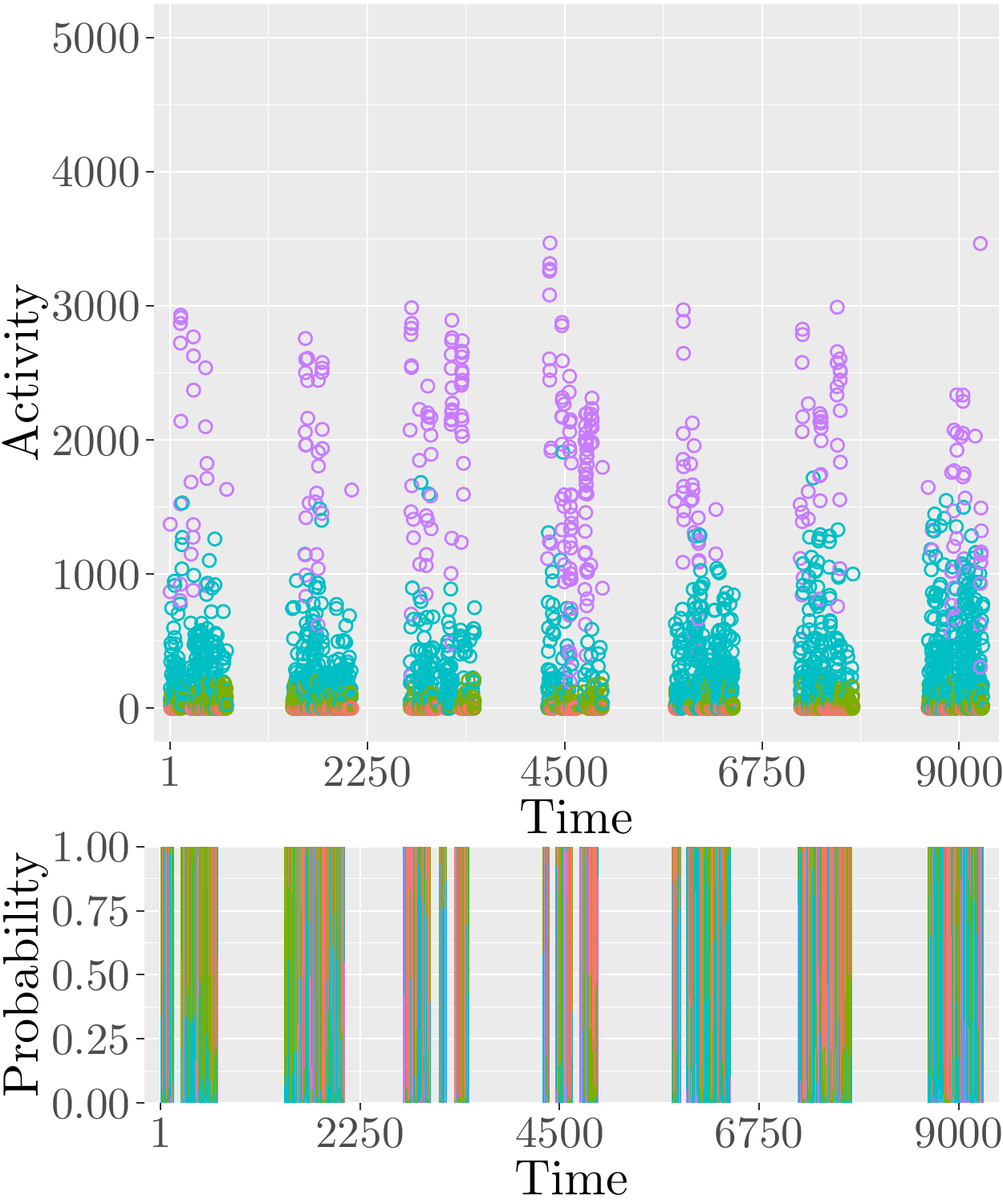}}
\subfigure[\emph{Sedentary} class]{\includegraphics[width=0.47\textwidth]{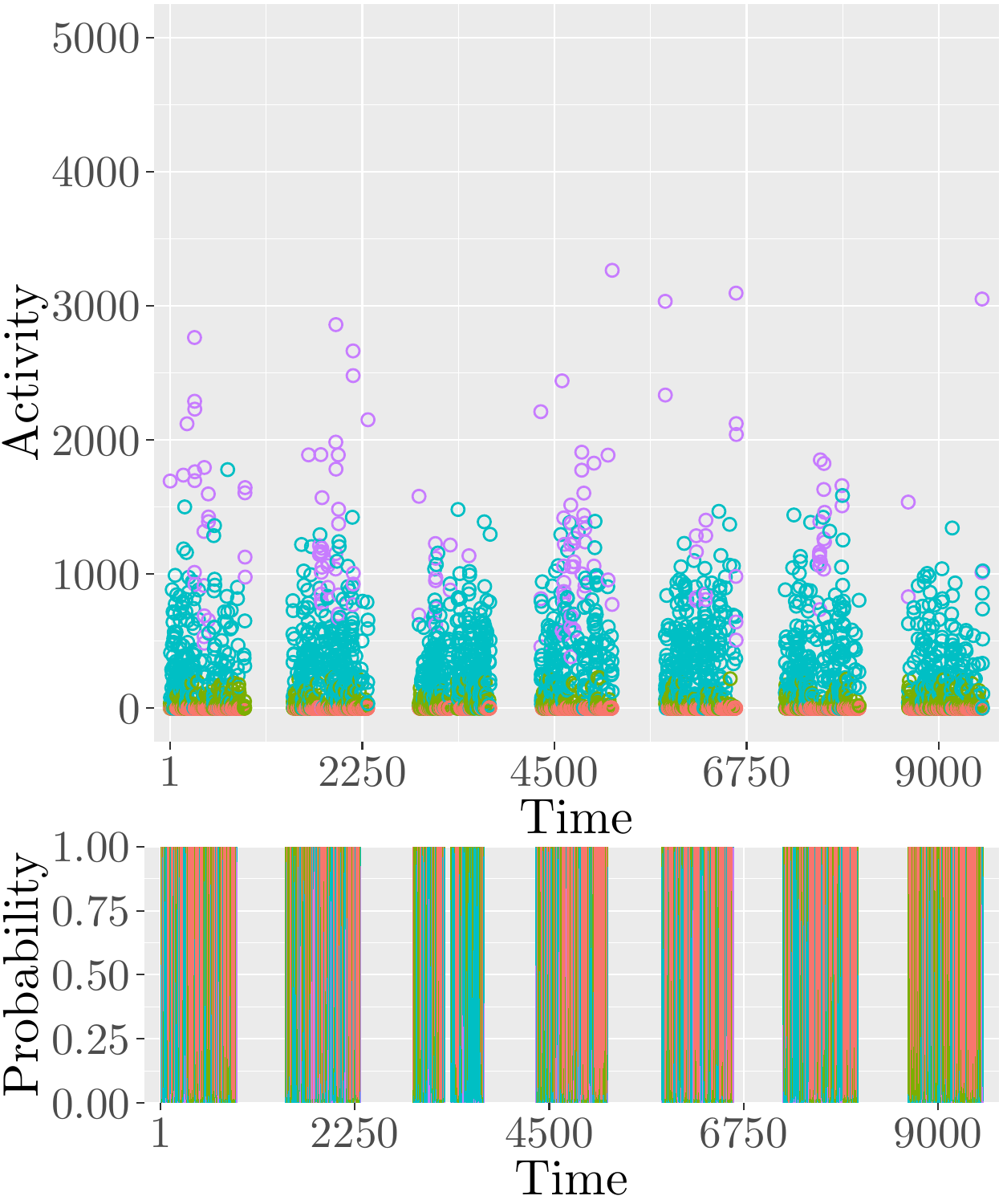}}
\end{center}
\begin{center}
\includegraphics[width=0.4\textwidth]{legfig4}
\end{center}
\caption{Examples of observation assigned into the five classes with the probabilities of the states.}
\label{fig:observations2}
\end{figure}

\begin{table}[H]
\caption{Transition matrix for the two remaining classes}
\label{tab:usual2}
\centering
\begin{tabular}{lcccc}
\hline
&\multicolumn{4}{c}{Class \emph{active}}\\
 & sleeping & low-level & moderate-level & intensive-level \\ 
 \hline
sleeping & 0.87 & 0.12 & 0.01 & 0.00 \\ 
low-level & 0.17 & 0.73 & 0.10 & 0.00 \\ 
moderate-level & 0.04 & 0.30 & 0.66 & 0.01 \\ 
intensive-level & 0.08 & 0.08 & 0.18 & 0.66 \\
\hline
&\multicolumn{4}{c}{Class \emph{sedentary}}\\
 & sleeping & low-level & moderate-level & intensive-level \\   
 \hline
sleeping & 0.79 & 0.16 & 0.05 & 0.00 \\ 
low-level & 0.17 & 0.66 & 0.16 & 0.01 \\ 
moderate-level & 0.05 & 0.14 & 0.79 & 0.03 \\ 
intensive-level & 0.01 & 0.02 & 0.15 & 0.82 \\ 
\end{tabular}
\end{table}

\bibliographystyle{apalike}
\bibliography{mabib.bib}
\end{document}